\pdfoutput=1
\documentclass{article}
\usepackage{stsi}

\title{Spatio-Temporal Stochastic Interventions for Causal Inference\\in Climate Science}
\author{Samuel Baugh\\Department of Statistics\\Pennsylvania State University}
\date{}

\newcommand{\Kdim}{50}
\newcommand{\KdimControl}{48}
\newcommand{\KdimPr}{20}
\newcommand{\nGrid}{10{,}368}
\newcommand{\nMonthsPi}{10{,}212}
\newcommand{\bootReplicates}{500}
\newcommand{\bootBlockYears}{10}
\newcommand{\nHistMembers}{46}
\newcommand{\nSspMembers}{15}
\newcommand{\kappaPresent}{10^{5}}
\newcommand{\emergenceYear}{1995}
\newcommand{\adjustInflateGlobalPct}{47}
\newcommand{\adjClassCropRevPct}{8}
\newcommand{\cropYear}{2014}
\newcommand{\adjClassCropRevMha}{125}
\newcommand{\adjClassCropPosRevMha}{60}
\newcommand{\adjClassCropNegRevMha}{65}

\newcommand{\cropOtherN}{eleven}
\newcommand{\cropOtherMha}{580}
\newcommand{\cropOtherRevMha}{5}
\newcommand{\cropWindowsRevPct}{61}
\newcommand{\vpdLandPresent}{0.196}
\newcommand{\vpdOceanPresent}{0.004}
\newcommand{\landAreaPct}{30}
\newcommand{\landSharePresentPct}{95}
\newcommand{\vpdMeanGlobal}{5.3}
\newcommand{\vpdSdAnnualGlobal}{0.061}
\newcommand{\vpdGlobalPresentSdAnnual}{1.0}
\newcommand{\vpdLandPresentSdAnnual}{1.2}
\newcommand{\essPresentJoint}{77}
\newcommand{\ccFixedLand}{0.28}
\newcommand{\ccConstRhLand}{0.13}
\newcommand{\ccConstRhOcean}{0.015}
\newcommand{\essPresentMarginal}{100}
\newcommand{\essFirstYear}{1951}
\newcommand{\essFirstJoint}{248}
\newcommand{\snrAreaPct}{24.0}
\newcommand{\snrAreaPosPct}{17.1}
\newcommand{\snrAreaNegPct}{6.9}
\newcommand{\sigmaHistNresid}{2{,}944}
\newcommand{\sigmaHistCondMin}{273}
\newcommand{\sigmaHistCondMax}{557}
\newcommand{\sigmaHistCondPenMin}{34}
\newcommand{\sigmaHistCondPenMax}{71}
\newcommand{\sspForcedTasMaxDiff}{0.22}
\newcommand{\ghgForcingObs}{0.443}
\newcommand{\ghgRmseForcingMax}{0.041}
\newcommand{\ghgRmseForcingMain}{0.028}
\newcommand{\basisForcedHybrid}{96.7}
\newcommand{\basisForcedControl}{69.9}
\newcommand{\basisSspHybrid}{99.2}
\newcommand{\basisSspControl}{71.6}
\newcommand{\basisForcedHistorical}{97.1}
\newcommand{\basisSspHistorical}{99.2}
\newcommand{\expOneZeroSelPct}{18.3}
\newcommand{\simExpTwoBiasMargHigh}{65}
\newcommand{\simExpTwoGammaMax}{5.5\times10^{-3}}
\newcommand{\simSteps}{700}
\newcommand{\simBurnIn}{200}
\newcommand{\simT}{500}
\newcommand{\simExpOneRecords}{300}
\newcommand{\simExpTwoRecords}{100}

\hypersetup{pdftitle={Spatio-Temporal Stochastic Interventions for Causal Inference in Climate Science}, pdfauthor={Samuel Baugh}, colorlinks, citecolor=blue, urlcolor=blue, linkcolor=blue, linktocpage=true}

\begin{document}

\maketitle

\begin{abstract}
Estimating causal effects in climate science, such as the effect of anthropogenic warming on crop loss, is challenging because of complex spatio-temporal dependence and the high-dimensional nature of the treatment.
To address this dependence and the resulting poor overlap between observed and counterfactual scenarios, we develop a spatio-temporal stochastic-intervention framework for estimating causal effects from climate observations.
We introduce a regularized estimator of the stochastic-intervention treatment effect that trades a controlled bias for a reduction in the weight variance caused by poor overlap.
Simulation studies show that this estimator attains lower mean squared error than alternative weighting estimators and removes the confounding bias of an unadjusted estimator.
We apply the framework to estimate the effect of historical warming on vapor-pressure deficit, a driver of crop stress, adjusting for precipitation, which confounds the effect by affecting both temperature and humidity.
In GISS-E2-1-G climate-model simulations, the global effect is distinguishable from zero in every year from \emergenceYear{} onward, and omitting the precipitation adjustment inflates the global estimate by \adjustInflateGlobalPct\%.
Adjustment reverses the sign of the estimate over \adjClassCropRevPct\% of global cropland (\adjClassCropRevMha{} million hectares), where an unadjusted analysis could misdirect adaptation between heat-focused and moisture-focused measures.
\end{abstract}

\tableofcontents
\newpage

\section{Introduction}

\subsection{Motivation}
Vapor-pressure deficit (VPD), a principal measure of atmospheric dryness, is a leading driver of increased wildfire activity and crop stress \citep{lobell2014greater, abatzoglou2016impact, williams2019observed}. Quantifying the extent to which observed changes in VPD are directly attributable to anthropogenic warming is a causal question, and since controlled experiments on the Earth's climate are impossible, this question must be answered from a combination of observations and climate-model simulations. 

In observations and in climate simulations, the direct effect of temperature on atmospheric dryness is confounded by moisture availability (for which we use precipitation as a proxy) because increased moisture generally decreases temperatures through latent cooling while also separately decreasing VPD through increased humidity. Separating the direct causal effect of temperature from its confounded association is important for informing adaptation efforts, as moisture-driven changes can be mitigated through increased irrigation and water storage, whereas temperature-driven changes require alternative approaches such as heat-tolerant crop varieties or shifted planting calendars \citep{challinor2014meta}. Failure to adjust for confounding may therefore result in suboptimal or misdirected adaptation strategies.

Confounding adjustment is a classical problem in causal inference with well-studied methods \citep{ding2024first}. Applying these methods to climate data is complicated, however, by spatial and temporal dependence resulting in violations of the stable unit treatment value assumption (SUTVA).
If the treatment unit were taken to be the state of the atmosphere at a particular location and time, SUTVA would be violated by both temporal carryover and spatial interference; temporal carryover arises because current treatments and outcomes depend on earlier system states by virtue of their evolution in time, and spatial interference arises because an outcome at one location depends on treatments elsewhere through physical processes including advection and mixing.
One way to satisfy SUTVA with climate data is to model the treatment unit as the entire spatial field (which handles spatial interference) and then model temporal carryover through assumptions on the causal directed acyclic graph (DAG).

While modeling the treatment unit as the entire spatial field allows for SUTVA to hold, this choice introduces challenges inherent in estimating causal effects with high-dimensional treatments. This dimensionality problem is particularly severe in the climate setting: even on a relatively coarse 2.5-degree grid, the global temperature field has more than ten thousand locations.
A key challenge is satisfying the \textit{overlap} assumption (also referred to as ``positivity''), which requires that treatment values possible under the intervention have positive density under the observational process. 
Even when overlap technically holds, in high-dimensional settings it often holds only weakly (meaning that the densities are close to zero) because discrepancies between treatment values accumulate across coordinates \citep{petersen2012diagnosing, diaz2012population, diaz2021nonparametric}.

Causal inference where the treatment is a continuous global field therefore requires a framework that addresses both spatio-temporal dependence and the resulting weak overlap between intervention and observational densities.
To address this need, we develop a framework consisting of a stochastic-intervention estimand defined as a conditional shift in the distribution of the treatment field, a procedure for estimating the observational and intervention densities, and a regularized inverse-probability-weighting estimator whose tunable penalty trades a controlled bias for a reduction in the weight variance caused by poor overlap; we then apply this framework to estimate the direct effect of warming on VPD from climate-model output.

\subsection{Prior Work}

The present work builds on the causal inference literature on stochastic interventions on the methodological side, and the scientific literature on the \textit{detection and attribution} of climate change on the application side.

On the methodological side, our work extends the stochastic-intervention approach, which generalizes the potential-outcomes framework to continuous treatments by defining an intervention as a change in the distribution of treatment values. In this setting, the causal estimand is defined by the contrast between expected potential outcomes under different intervention distributions, as opposed to the contrast between two fixed treatment values.
This approach was formalized as a population intervention effect by \cite{diaz2012population} and has since been extended to represent interventions in time series and other structured settings.
\cite{kennedy2019nonparametric} and \cite{diaz2021nonparametric} study longitudinal interventions under time-varying confounding, while \cite{wu2024assessing} develop an incremental propensity intervention for multisite time series with binary treatments.

\cite{papadogeorgou2022causal} extend stochastic interventions to spatio-temporal point-process treatments for studying airstrikes and insurgent violence in Iraq.
They define interventions by modifying the intensity of the treatment point process, and address temporal confounding by using weights conditioned on treatment history.
Although the treatment is a spatial pattern, each intervention perturbs the intensity through one or two scalar parameters, which keeps the intervention distributions close to the observational distribution and facilitates overlap.
This approach does not immediately apply to continuous spatial fields, where the joint distribution also depends on a covariance structure that captures dependence across locations.

Other approaches for causal inference on spatial and environmental data have largely focused on spatial confounding and interference with discrete, threshold-based, or scalar treatments at individual locations \citep{reich2021review}.
One set of methods for handling unmeasured spatial confounding is spectral adjustment, which projects scalar treatments and outcomes onto a Fourier basis to separate the causal signal from spatially smooth confounding \citep{guan2022spectral}. Distance-adjusted propensity score matching addresses the same problem within the matching framework \citep{papadogeorgou2019adjusting}.
Interference through physical dynamics has also been modeled directly by \cite{wikle2024causal}, who represent atmospheric transport as an explicit interference mechanism in assessing the health impacts of power-plant emissions.
\cite{christiansen2022toward} develop a framework for causal inference with spatio-temporal data under latent confounding, estimating the effect of a time-varying treatment on a spatial outcome field and validating it on the effect of armed conflict on forest loss in Colombia; their treatment, however, consists of binary indicators rather than a continuous-valued field.

The approaches discussed above deal with spatio-temporal and high-dimensional data either by modeling the treatment as a scalar or by restricting the intervention a priori to a lower-dimensional space.
Other approaches exist for multidimensional treatments, although they are generally not applied to spatial treatment fields.
\cite{nabi2022semiparametric} develop causal sufficient dimension reduction for multidimensional treatments, identifying linear combinations of the treatment vector that preserve the causal relationship with the outcome.
Their approach, applied to radiation dose-volume vectors with up to twenty dimensions, models the response surface as a function of the full treatment vector; this is feasible at moderate dimensions but intractable when the treatment is a spatial field with more than ten thousand components.
\cite{jenkins2025causal} handle high-dimensional treatments in a radiotherapy setting through sparsity assumptions; such assumptions, while natural for genetic or clinical covariates, are less natural for spatially correlated treatment fields in which all locations may contribute to the causal effect.
Our approach is instead driven by the covariance structure of the treatment distribution, trading the optimality of outcome-driven reduction for computational tractability at the dimensions of global climate fields.

On the application side, our research question is most closely related to the climate science literature on detection and attribution. Detection and attribution methods aim to quantify the extent to which the response to external forcing (for example, human-driven changes in greenhouse gas or aerosol concentrations) can be \emph{detected} in the observational record, as well as the extent to which observed changes can be \emph{attributed} to that forcing \citep{hammerling2019climate, masson2021ipcc}.
The most widely used technique, optimal fingerprinting, regresses the observed pattern of a climate variable onto forced-response patterns (``fingerprints'') estimated from climate-model simulations and treats the residual as internal variability \citep{hasselmann1993optimal, hegerl1996detecting, allen2003estimating}.

Although the questions asked by detection and attribution are causal, the field developed largely independently of causal inference in statistics. Detection and attribution has since been placed in the structural causal model framework of Pearl \citep{hannart2016causal, hannart2018probabilities} and analyzed through Granger causality \citep{risser2025granger}, but formulations in the potential-outcomes framework, as used in our approach, remain uncommon. Our approach also differs in the definition of the treatment; in detection and attribution, the treatment is the external forcing itself, and simulations with and without the forcing are taken to represent potential outcomes. Because forcing alters temperature and precipitation together, such a comparison cannot separate the effect of warming on VPD from the confounding effect of precipitation. We instead take the treatment to be the climate field that results from the forcing, which further allows for the estimation of treatment effects from a single realization of the climate system. While we use climate-model output as surrogate observations in the present application, in principle our approach requires climate-model simulations only to estimate the distributions that define the stochastic interventions, as opposed to using simulations to represent the potential outcomes themselves.

\subsection{Contributions and Paper Structure}

The contributions of the present work can be divided into four areas.
First, we define a stochastic-intervention estimand for spatio-temporal systems in which the treatment is the entire spatial field, and show that this estimand is identified through an inverse-probability-weighting expression (Section~\ref{sec:assumptions} and Supplementary Proposition~\ref{prop:identification}).
Second, we propose a regularized estimator whose weights are computed from projected Gaussian propensity models and regularized by a tunable penalty term that balances the bias--variance tradeoff, provide a penalty-selection criterion that uses only the available data, and develop a block-bootstrap procedure for estimating standard errors (Section~\ref{sec:estimation}).
Third, in synthetic-data studies constructed from climate-model output with known ground truth, we show that the regularized estimator attains lower error than ordinary, truncated, and weight-capped inverse probability weighting, and that the estimator correctly adjusts for confounding bias (Section~\ref{sec:simulation}).
Fourth, we estimate the effect of historical forced temperature change on VPD, adjusted for precipitation, from GISS-E2-1-G simulations; the 95\% block-bootstrap interval for the global-mean effect excludes zero in every intervention year from \emergenceYear{} onward, omitting the precipitation adjustment inflates the global estimate by \adjustInflateGlobalPct\%, and adjustment reverses the sign of the estimate over \adjClassCropRevPct\% of global cropland, with consequences for the choice between heat-focused and moisture-focused adaptation in such regions (Section~\ref{sec:application}).

The remainder of the paper is organized as follows.
Section~\ref{sec:setting} describes the causal structure, the intervention model, and the estimand.
Section~\ref{sec:estimation} states the identifying assumptions and develops the regularized estimator, the penalty-selection rule, and a block-bootstrap procedure for uncertainty quantification.
Section~\ref{sec:simulation} presents the simulation studies, and Section~\ref{sec:application} presents the application and its implications for adaptation.
Section~\ref{sec:discussion} discusses limitations and future extensions to doubly robust estimators, single-forcing experiments, and observational reanalysis data.
The Supplementary Material contains the identification proof, the simulation-study details, and application details and diagnostics, including the cropland mask and the regional cropland breakdown.

\section{Setting and Estimand}\label{sec:setting}

\subsection{Causal Structure and Potential Outcomes}

Let $s_1, \ldots, s_n \in \mathcal{S} \subset \mathbb{R}^2$ denote $n$ spatial locations and let $T$ denote the number of time periods. At each time $t$, let $Z_t(s_i), X_t(s_i), Y_t(s_i) \in \mathbb{R}$ denote the treatment, confounder, and outcome fields, respectively, at location $s_i$. Let $\mathbf{Z}_t, \mathbf{X}_t, \mathbf{Y}_t \in \mathbb{R}^n$ denote vectors consisting of the values of the corresponding fields.
These variables are related through the causal structure illustrated in Figure~\ref{fig:dag}, where nodes represent entire spatial fields, and arrows encode causal relations between fields.

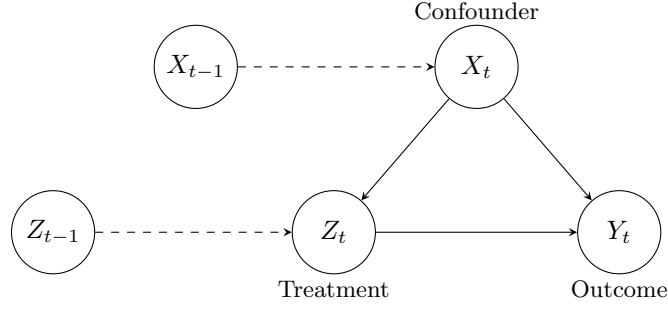
\begin{figure}[tbp]
\centering
\begin{tikzpicture}[>=stealth,
    every node/.style={circle, draw, minimum size=1.1cm, inner sep=0pt},
    every label/.style={rectangle, draw=none, minimum size=0pt, inner sep=2pt}]
  \node[label=above:{\small Confounder}] (X) {$X_t$};
  \node[below left=1.4cm and 1.1cm of X, label=below:{\small Treatment}] (Z) {$Z_t$};
  \node[below right=1.4cm and 1.1cm of X, label=below:{\small Outcome}] (Y) {$Y_t$};
  \node[left=2.6cm of X] (Xp) {$X_{t-1}$};
  \node[left=2.6cm of Z] (Zp) {$Z_{t-1}$};
  \draw[->] (X) -- (Z);
  \draw[->] (X) -- (Y);
  \draw[->] (Z) -- (Y);
  \draw[->, dashed] (Xp) -- (X);
  \draw[->, dashed] (Zp) -- (Z);
\end{tikzpicture}
\caption{Directed acyclic graph showing the causal structure.
  Within time period $t$, the variables form a classical confounding triangle, where
  the confounder $X_t$ is a common cause of both the treatment $Z_t$ and the
  outcome $Y_t$.
  Dashed arrows indicate temporal carryover from $t{-}1$ to $t$.}
\label{fig:dag}
\end{figure}

Following the potential-outcomes framework \citep{rubin1974estimating}, let $Y_t(s; \mathbf{Z})$ denote the potential outcome at location $s$ and time $t$ under treatment field $\mathbf{Z}$. In the classical causal inference setting, the estimand is defined as the contrast between potential outcomes at two fixed treatment values. Given the continuous nature of our treatment, we instead adopt the stochastic-intervention framework, where the causal estimand is defined as a contrast between two treatment distributions; these interventions are defined in the following subsection. In the absence of an intervention, the process defined by $X_t$, $Y_t$, and $Z_t$ evolves naturally according to the causal structure in Figure~\ref{fig:dag}; let $f_{\mathrm{obs}}$ denote the corresponding observational distribution of the treatment. 

For inference on the direct effect of anthropogenic warming on aridity, we take near-surface air temperature ($Z_t=\mathrm{tas}_t$) as the treatment, precipitation ($X_t=\mathrm{pr}_t$) as the confounder, vapor-pressure deficit ($Y_t=\mathrm{VPD}_t$) as the outcome, and months as the time unit $t$. We interpret VPD, defined as the difference between saturation vapor pressure at a given temperature and actual vapor pressure, as a measure of aridity because it reflects atmospheric drying relevant to crop stress and wildfire risk.
This three-variable causal structure represents the dominant physical pathways at monthly scales, where precipitation affects temperature through clouds and surface-energy exchange, temperature affects VPD through the Clausius--Clapeyron relation, and precipitation affects VPD through moisture transport and evaporation \citep{masson2021ipcc}. We note that the direction of the arrows in the causal structure depends on the temporal resolution; for the monthly time units used here, we assume that the reverse effect of temperature on precipitation is negligible after monthly averaging, and that precipitation evolves exogenously from temperature and humidity at monthly lags. 

\subsection{The Intervention Model and Estimand}\label{sec:intervention}\label{sec:one-step}

To accommodate the regularized estimator of Section~\ref{sec:propensity}, we define the intervention model and estimand in projected coordinates.
Let $\Vz = [\mathbf{v}_1, \ldots, \mathbf{v}_{\Kz}] \in \mathbb{R}^{n \times \Kz}$ denote the treatment basis, and define the projected treatment scores at time $t$ as
\begin{equation*}
\mathbf{z}_t = \Vz^\top \mathbf{Z}_t \in \mathbb{R}^{\Kz}.
\end{equation*}
We analogously project the confounder field using a basis matrix $\Vx \in \mathbb{R}^{n \times \Kx}$:
\begin{equation*}
\mathbf{x}_t = \Vx^\top \mathbf{X}_t \in \mathbb{R}^{\Kx}.
\end{equation*}
In the experiments and application, the bases are constructed from the leading spatial patterns of the treatment and confounder fields; however, the general framework is flexible and permits alternative choices.
We assume that potential outcomes depend on the treatment field only through its projected scores, so that $Y_t(s;\mathbf{z}_t)$ can be used in place of $Y_t(s;\mathbf{Z}_t)$. 
Additionally, let $\pi_{x,z}$ denote the joint distribution of the projected lagged treatment $\mathbf{z}_{t-1}$ and the projected confounder $\mathbf{x}_t$.

To represent interventions on the causal structure of Figure~\ref{fig:dag}, we model the intervention distributions as Gaussian vector autoregressions on the lagged treatment and the current confounder:
\begin{equation}\label{eq:intervention-dist}
f_{r;t^{*}}(\mathbf{z}_t \mid \mathbf{z}_{t-1}, \mathbf{x}_t) = \mathcal{N}\!\left(\mathbf{z}_t;\; \boldsymbol{\Phi}_{r} \mathbf{z}_{t-1} + \PhiTwo \mathbf{x}_t + \boldsymbol{\mu}_{r}(t^{*}),\; \boldsymbol{\Sigma}_{\varepsilon,r}\right),
\end{equation}
where $r\in\{0,1\}$ indexes the two intervention distributions that will be used in the causal contrast, $\boldsymbol{\Phi}_{r}$ represents the temporal persistence of the treatment, $\PhiTwo$ describes the effect of the measured confounder on the treatment, $\boldsymbol{\mu}_{r}(\cdot)$ is an intercept trajectory, $\boldsymbol{\Sigma}_{\varepsilon,r}$ is the residual covariance, and $t^*$ is an additional time parameter permitting temporal non-stationarity in the intervention distribution.
Each intervention distribution is parameterized by $(\boldsymbol{\Phi}_{r}, \boldsymbol{\mu}_{r}(\cdot), \boldsymbol{\Sigma}_{\varepsilon,r})$, while the confounder coefficient $\PhiTwo$ is shared by both $f_0$ and $f_1$.
Note that this definition of the intervention is conditioned on the confounder $\mathbf{x}_t$ to represent the effect of intervening on the dynamics of the process; in our application, anthropogenic influence intervenes on the growth of temperature over time conditional on the value of precipitation.

In the context of the application, $f_0$ and $f_1$ represent \textit{control} and \textit{forced} distributions; specifically, $f_0$ represents the natural climate distribution prior to anthropogenic influence, whereas $f_1$ represents anthropogenic influence driven by increasing greenhouse gas concentrations and other atmospheric pollutants, which we approximate by historical simulations that also include smaller natural influences. With this definition of $f_1$, the time index $t^*$ in Equation~\eqref{eq:intervention-dist} represents forcing at its historical level in year $t^*$, which in the application ranges from 1951 to 2025. We assume that the control intervention $f_0$ is stationary, so when referring to $f_0$ we will drop the dependence on $t^*$.

Write $\mathbf{c}_t = (\mathbf{z}_{t-1}, \mathbf{x}_t)$ for the conditioning pair of the lagged treatment and the current confounder. The mean potential outcome under intervention $r$ is
\begin{equation}\label{eq:psi-def}
    \psi_{r;t^{*}}(s) = \mathbb{E}_{\mathbf{c}_t \sim \pi_{x,z}}\Big\{\mathbb{E}_{\mathbf{z} \sim f_{r;t^{*}}(\cdot \mid \mathbf{c}_t)}\big[Y_t(s;\mathbf{z}) \mid \mathbf{c}_t\big]\Big\},
\end{equation}
where the inner expectation is over the potential outcome given $\mathbf{c}_t$ and over a treatment value $\mathbf{z}$ drawn from the intervention density independently of the potential outcomes given $\mathbf{c}_t$, and the outer expectation averages over the joint distribution $\pi_{x,z}$ of the conditioning pair.
The estimand is then defined as the contrast between the mean potential outcomes under $f_1$ and $f_0$:
\begin{equation}\label{eq:one-step-averaged}
    \tau_{t^{*}}(s) = \psi_{1;t^{*}}(s) - \psi_{0}(s),
\end{equation}
where $\psi_{0}(s)$ carries no intervention-time index because $f_0$ is assumed to be temporally stationary.

\section{Estimation}\label{sec:estimation}

Our estimator for Equation~\eqref{eq:one-step-averaged} is built on inverse probability weighting (IPW), which adjusts for confounding by weighting observations with the reciprocal of the observational treatment density \citep{ding2024first}. Extensions to doubly robust estimators, which combine IPW with outcome regression for additional theoretical guarantees, are discussed in Section~\ref{sec:discussion} and reserved for future work.

A central difficulty in applying IPW to global spatial fields is the problem of overlap, which becomes more restrictive as the treatment dimension grows \citep{damour2021overlap}; in particular, when $f_r$ and $f_{\mathrm{obs}}$ are high-dimensional, standard IPW weights will tend to collapse onto a few observations, resulting in estimators with high variance and low effective sample size (ESS; \citealp{bengtsson2008curse, snyder2008obstacles}).
We address this weight collapse through an adaptive regularization technique (Section~\ref{sec:propensity}) that controls the bias--variance tradeoff through the addition of a penalty term to the residual covariance matrices.

\subsection{Assumptions}\label{sec:assumptions}

Our estimation procedure requires four causal assumptions, which we state in terms of the projected scores defined in Section~\ref{sec:intervention}.
The first three (consistency, unconfoundedness, and positivity) are the classical identification conditions of observational causal inference \citep{rosenbaum1983central, ding2024first}.
The fourth assumption is the invariance of auxiliary mechanisms, which requires that the intervention act on the joint distribution only through the treatment distribution $f_r$; versions of this assumption are standard for stochastic interventions in general \citep{diaz2012population} and for modified treatment policies in particular \citep{haneuse2013estimation, diaz2021nonparametric}.
Throughout, $Y_t(s; \mathbf{z})$ denotes the potential outcome at location $s$ under projected treatment score $\mathbf{z}$, and $\mathcal{F}_{t-1} = \sigma\big(\{(\mathbf{z}_j, \mathbf{x}_j, \mathbf{Y}_j)\}_{j \le t-1}\big)$ is the observed history through time $t-1$.

\begin{assumption}[Consistency]\label{asm:consist}
    If $\mathbf{z}_t$ is drawn from $f_{\mathrm{obs}}$, $f_0$, or $f_{1;t^{*}}$, and $\mathbf{z}_t=\mathbf{z}$, then $Y_t(s) = Y_t(s; \mathbf{z})$.
\end{assumption}

\begin{assumption}[Unconfoundedness]\label{asm:exch}
For all $\mathbf{z}$, $s$, $t$: $Y_t(s; \mathbf{z}) \perp \mathbf{z}_t \mid \mathbf{z}_{t-1}, \mathbf{x}_t$.
\end{assumption}

\begin{assumption}[Positivity/Overlap]\label{asm:pos}
    For all $\mathbf{z}$ in the support of $f_{r;t^{*}}(\cdot \mid \mathbf{z}_{t-1}, \mathbf{x}_t)$, and all $(\mathbf{z}_{t-1}, \mathbf{x}_t)$ in the support of $\pi_{x,z}(\cdot,\cdot)$, $f_{\mathrm{obs}}(\mathbf{z} \mid \mathbf{z}_{t-1}, \mathbf{x}_t) > 0$.
\end{assumption}

\begin{assumption}[Invariance of Auxiliary Mechanisms]\label{asm:auxiliary}
    The joint distribution of the confounder, treatment, and outcome at time $t$, given the history $\mathcal{F}_{t-1}$, under each of $f\in\{f_{\mathrm{obs}},f_0,f_{1;t^{*}}\}$, factors as
\begin{equation*}
f(\mathbf{x}_t, \mathbf{z}_t, \mathbf{Y}_t \mid \mathcal{F}_{t-1}) = f_X(\mathbf{x}_t \mid \mathcal{F}_{t-1})\, f(\mathbf{z}_t \mid \mathbf{z}_{t-1}, \mathbf{x}_t)\, f_Y(\mathbf{Y}_t \mid \mathbf{z}_t, \mathbf{x}_t, \mathcal{F}_{t-1}),
\end{equation*}
where the confounder factor $f_X$ and the outcome factor $f_Y$ are identical for each $f$.
\end{assumption}

Under Assumptions~\ref{asm:consist}--\ref{asm:auxiliary}, stationarity of the observational process, and a mild integrability condition, the mean potential outcome $\psi_{r;t^{*}}(s)$ of Equation~\eqref{eq:psi-def} is identified by the inverse-probability-weighting expression
\begin{equation}\label{eq:identification}
\psi_{r;t^{*}}(s) = \mathbb{E}_{f_{\mathrm{obs}}}\!\left[\frac{f_{r;t^{*}}(\mathbf{z}_t \mid \mathbf{z}_{t-1}, \mathbf{x}_t)}{f_{\mathrm{obs}}(\mathbf{z}_t \mid \mathbf{z}_{t-1}, \mathbf{x}_t)} \cdot Y_t(s)\right],
\end{equation}
which reweights the observational sample by the ratio of the intervention to the observational density; the proof is given in Supplementary Section~\ref{app:identification}.
The causal estimand of Equation~\eqref{eq:one-step-averaged} is then identified as the contrast of Equation~\eqref{eq:identification} over $r\in\{0,1\}$. 

\subsection{Propensity Model and Estimator}\label{sec:propensity}

The observational density $f_{\mathrm{obs}}$ takes the role of the propensity model in the identification expression~\eqref{eq:identification}. To model $f_{\mathrm{obs}}$, we use a Gaussian autoregressive process analogous to the intervention distributions of Equation~\eqref{eq:intervention-dist}:
\begin{equation*}\label{eq:var1}
f_{\mathrm{obs}}(\mathbf{z}_t \mid \mathbf{z}_{t-1}, \mathbf{x}_t) = \mathcal{N}\!\left(\mathbf{z}_t;\; \boldsymbol{\Phi}_{\mathrm{obs}} \mathbf{z}_{t-1} + \PhiTwo \mathbf{x}_t + \boldsymbol{\mu}_{\mathrm{obs}},\; \boldsymbol{\Sigma}_{\varepsilon,\mathrm{obs}}\right),
\end{equation*}
with temporal persistence term $\boldsymbol{\Phi}_{\mathrm{obs}}$, confounding term $\PhiTwo$, intercept $\boldsymbol{\mu}_{\mathrm{obs}}$, and residual covariance $\boldsymbol{\Sigma}_{\varepsilon,\mathrm{obs}}$. For the intervention distribution $f_{r;t^*}$, the weight in the identification expression~\eqref{eq:identification} is the conditional density ratio
\begin{equation}\label{eq:weight-def-setting}
w_t^{(r;t^*)} = \frac{f_{r;t^{*}}(\mathbf{z}_t \mid \mathbf{z}_{t-1}, \mathbf{x}_t)}{f_{\mathrm{obs}}(\mathbf{z}_t \mid \mathbf{z}_{t-1}, \mathbf{x}_t)}.
\end{equation}

Unlike truncation-based approaches, which are sensitive to the choice of a small number of leading components \citep{allen1999checking}, we propose an adaptive regularization approach that retains all $\Kz$ projected components and adjusts the estimated residual covariance matrices by a penalty term.
For each $r \in \{0, 1\}$, we define the regularized covariance as
\begin{equation*}\label{eq:regularized-sigma}
\boldsymbol{\Sigma}_{\varepsilon,r}^{(\kappa)} = \boldsymbol{\Sigma}_{\varepsilon,r} + \mathrm{diag}\!\left(\frac{\kappa}{\lambda_1}, \ldots, \frac{\kappa}{\lambda_{\Kz}}\right)
\end{equation*}
(and analogously for $\boldsymbol{\Sigma}_{\varepsilon,\mathrm{obs}}$), where $\kappa \ge 0$ is the regularization parameter and $\lambda_k = \mathbf{v}_k^\top \operatorname{Cov}_{\mathrm{obs}}(\mathbf{Z}_t)\, \mathbf{v}_k$ is the projected observational variance of the treatment along component $k$.
When the residual covariance is diagonal in the penalty coordinates with entries $\sigma_k^2$, the penalty multiplies the precision of component $k$ by the factor $\sigma_k^2/(\sigma_k^2 + \kappa/\lambda_k)$, which most strongly reduces the precision of low-variance components, and hence their contribution to the weights, while leaving the precision of high-variance components nearly unchanged.
Since both densities in~\eqref{eq:weight-def-setting} are $\Kz$-dimensional Gaussians sharing the same projected basis $\Vz$, the regularized log-weight at each time step can then be written as
\begin{equation}\label{eq:logw-regularized-general}
\begin{split}
\log w_t^{(r;t^*)}(\kappa) = -\frac{1}{2}\bigg[&\log\frac{|\boldsymbol{\Sigma}_{\varepsilon,r}^{(\kappa)}|}{|\boldsymbol{\Sigma}_{\varepsilon,\mathrm{obs}}^{(\kappa)}|} + (\mathbf{z}_t - \mathbf{m}_t^{(r;t^*)})^\top (\boldsymbol{\Sigma}_{\varepsilon,r}^{(\kappa)})^{-1} (\mathbf{z}_t - \mathbf{m}_t^{(r;t^*)}) \\
&- (\mathbf{z}_t - \mathbf{m}_t^{(\mathrm{obs})})^\top (\boldsymbol{\Sigma}_{\varepsilon,\mathrm{obs}}^{(\kappa)})^{-1} (\mathbf{z}_t - \mathbf{m}_t^{(\mathrm{obs})})\bigg],
\end{split}
\end{equation}
where $\mathbf{m}_t^{(r;t^*)}$ and $\mathbf{m}_t^{(\mathrm{obs})}$ are the conditional means under $f_{r;t^{*}}$ and $f_\mathrm{obs}$, respectively, and we write the weights $w_t^{(r;t^*)}(\kappa)$ as a function of $\kappa$ to make this dependence explicit.
When $\kappa = 0$, \eqref{eq:logw-regularized-general} is the logarithm of the unpenalized $\Kz$-dimensional density ratio~\eqref{eq:weight-def-setting}; as $\kappa \to \infty$, the penalty dominates and the weights approach uniformity. As $\kappa$ grows, every term of~\eqref{eq:logw-regularized-general} shrinks toward zero, but the contributions of low-variance components shrink first, so the weights lose their sensitivity to those directions before the high-variance ones.

We estimate the mean potential outcome $\psi_{r;t^{*}}(s)$ through the H\'{a}jek estimator \citep{hajek1971comment} with the regularized weights:
\begin{equation*}\label{eq:hajek}
\hat{\psi}_{r;t^{*}}(s;\kappa) = \frac{\sum_{t=1}^T w_t^{(r;t^*)}(\kappa) \, Y_t(s)}{\sum_{t=1}^T w_t^{(r;t^*)}(\kappa)}.
\end{equation*}
We then define the \textit{$\kappa$-regularized estimator} for the treatment contrast as
\begin{equation}
    \widehat{\tau}_{t^{*}}(s;\kappa) = \hat{\psi}_{1;t^{*}}(s;\kappa) - \hat{\psi}_{0}(s;\kappa).\label{eq:kappa-regularized}
\end{equation}  We note that since $f_0$ is assumed to be stationary over time, the $\hat{\psi}_{0}(s;\kappa)$ term does not depend on an intervention-time index $t^*$.

\subsection{Penalty Selection}\label{sec:kappa-selection}

The penalty term reduces variance through controlling weight concentration at the cost of introducing bias, and a natural selection rule for $\kappa$ would minimize the root mean squared error (RMSE) of the estimator evaluated against a target value. The ground truth for the causal estimand of Equation~\eqref{eq:one-step-averaged} is unknown, but the selection rule can exploit the fact that the effect of the intervention on the treatment itself is known by construction.
The identification expression~\eqref{eq:identification} holds with $Y_t(s)$ replaced by any function of $\mathbf{z}_t$; taking the reconstructed treatment field $Z_t(s) = [\Vz\mathbf{z}_t](s)$ as the outcome, and using $\mathbb{E}_{f_{r;t^{*}}}[\mathbf{z}_t \mid \mathbf{z}_{t-1}, \mathbf{x}_t] = \mathbf{m}_t^{(r;t^*)}$, yields the treatment contrast
\begin{equation*}\label{eq:treatment-contrast}
\tau^{Z}_{t^{*}}(s) = \Big[\Vz\, \mathbb{E}_{(\mathbf{z}_{t-1}, \mathbf{x}_t) \sim \pi_{x,z}}\big(\mathbf{m}_t^{(1;t^*)} - \mathbf{m}_t^{(0)}\big)\Big](s).
\end{equation*}
This contrast can be estimated in two ways: directly, by averaging the fitted conditional mean shift over the observed pairs $(\mathbf{z}_{t-1}, \mathbf{x}_t)$:
\begin{equation*}
\widehat{\tau}^{Z}_{t^{*}}(s) = \frac{1}{T}\sum_{t=1}^{T}\big[\Vz\big(\mathbf{m}_t^{(1;t^*)} - \mathbf{m}_t^{(0)}\big)\big](s),
\end{equation*}
or by the $\kappa$-regularized estimator~\eqref{eq:kappa-regularized} with $Z_t(s)$ as the outcome, which we denote $\widehat{\tau}^{Z}_{t^{*}}(s;\kappa)$.
Because the direct estimator does not involve weighting or regularization, we use it as the target value for selecting the RMSE-minimizing value of $\kappa$.
Formally, let $\mathcal{T}$ denote a set of intervention times for evaluating the forced distribution, and let
\begin{equation*}
 D(\kappa;t^{*})=\sum_s a_s\big\{\widehat{\tau}^{Z}_{t^{*}}(s;\kappa) -\widehat{\tau}^{Z}_{t^{*}}(s)\big\}^2
\end{equation*}
denote the area-weighted squared error at $t^*\in \mathcal{T}$, with area weights $a_s$. Then we estimate $\kappa$ through
\begin{equation}\label{eq:mse-selection}
 \widehat\kappa=\arg\min_{\kappa\in\mathcal{K}}
 \frac{1}{|\mathcal{T}|}\sum_{t^{*}\in\mathcal{T}}
 \frac{D(\kappa;t^{*})}{\min_{\kappa'\in\mathcal{K}}D(\kappa';t^{*})},
\end{equation}
where $\mathcal{K}$ is a nonnegative grid containing $\kappa=0$.
The set $\mathcal{T}$ can be a singleton, and averaging over intervention times is unnecessary if both intervention distributions are stationary. The bias--variance tradeoff is illustrated in Figure~\ref{fig:bias-variance}, using synthetic data from the study in Section~\ref{sec:sim-comparison} with a known ground truth.

\begin{figure}[tbp]
\centering
\includegraphics[width=\textwidth]{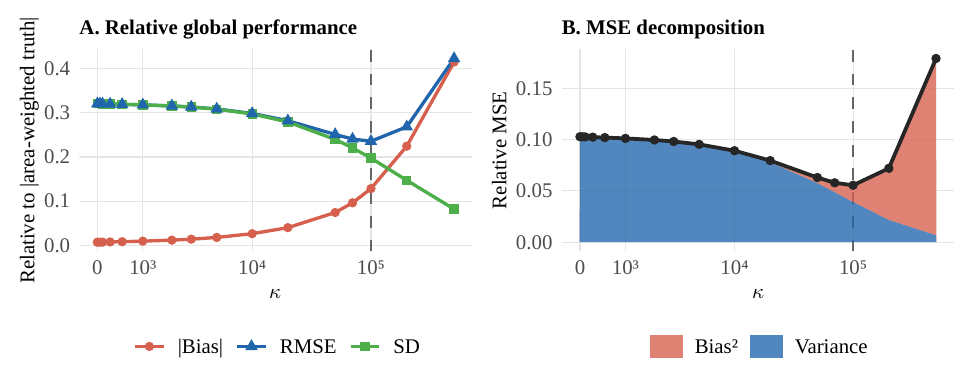}
\caption{Bias--variance tradeoff for $\widehat{\tau}(s;\kappa)$, evaluated using synthetic data with known ground truth from the experiment of Section~\ref{sec:sim-comparison}. Panel~A shows absolute bias, sampling standard deviation, and RMSE of the global estimator, divided by the absolute global truth. Panel~B shows variance and squared bias divided by the squared global truth, with their sum shown in black. The dashed vertical line marks the grid value of $\kappa$ minimizing simulated global RMSE.}
\label{fig:bias-variance}
\end{figure}

\subsection{Uncertainty Quantification}\label{sec:bootstrap}

To estimate the standard error of the estimator $\widehat{\tau}_{t^{*}}(s;\kappa)$, we use a temporal block bootstrap method.
This procedure samples whole fields, paired with their treatment, confounder, and outcome values, in temporally contiguous blocks. 
Each replicate draws $\lceil T/L \rceil$ block starting points independently and uniformly from the record and concatenates blocks of length $L$; temporal transitions that cross block boundaries are discarded. On each replicate, the observational propensity model and the directional variance parameters $\lambda_k$ are refit, and the estimator~\eqref{eq:kappa-regularized} is recomputed.
In the application, where time points are months, these blocks are taken to consist of complete years, so that $L$ is a multiple of twelve, and the estimator normalizes within calendar month and averages the monthly contrasts equally (Section~\ref{sec:app-estimation}).
For each replicate, the penalty is held at the value selected using Equation~\eqref{eq:mse-selection}.

From $B$ bootstrap replicates $\{\widehat{\tau}_{t^{*}}^{(b)}(s)\}_{b=1}^{B}$, the bootstrap standard error is then estimated by
\begin{equation*}
\widehat{\mathrm{SE}}(s) = \left[\frac{1}{B-1}\sum_{b=1}^{B}\Big(\widehat{\tau}_{t^{*}}^{(b)}(s) - \textstyle\frac{1}{B}\sum_{b'=1}^{B}\widehat{\tau}_{t^{*}}^{(b')}(s)\Big)^2\right]^{1/2}.
\end{equation*}
In addition to standard error estimates, the bootstrap distribution can be used to construct percentile intervals directly.
In the application, we say that the effect has emerged by intervention time $t^{*}$ when the 95\% percentile interval excludes zero at $t^{*}$ and at every later intervention time considered, and we call the ratio of an estimate to its bootstrap standard error its signal-to-noise ratio. 
In the application, we use $B=\bootReplicates$ replicates with blocks of \bootBlockYears{} years, so that $L=12\cdot\bootBlockYears$ months.

\section{Simulation Studies}\label{sec:simulation}

To validate the estimator against a known ground truth, we simulate data from semi-synthetic structural causal models whose statistical properties are taken from general circulation model (GCM) simulations.
Two experiments are reported below.
The first evaluates the penalty-selection rule of Equation~\eqref{eq:mse-selection} and compares the resulting estimator with truncation and weight-capping alternatives (Section~\ref{sec:sim-comparison}).
The second evaluates the ability of the estimator to correctly adjust for confounding of known strength (Section~\ref{sec:sim-confounding}).
We first describe the data used in the semi-synthetic simulations and the shared simulation framework.

\subsection{Data and Simulation Framework}\label{sec:sim-data}

To generate synthetic data with realistic climate properties, we use output from GISS-E2-1-G, a coupled atmosphere--ocean general circulation model contributed to Phase 6 of the Coupled Model Intercomparison Project (CMIP6) by the NASA Goddard Institute for Space Studies \citep{kelley2020giss}; output from this climate model is also used in the application in Section~\ref{sec:application}.
All fields are regridded to a common global $2.5^\circ \times 2.5^\circ$ latitude--longitude analysis grid from the native atmospheric grid of $2^\circ \times 2.5^\circ$, yielding $n = \nGrid$ spatial gridcells at $144 \times 72$ resolution.
For the simulation studies, we obtain the basis representation of temperature from annual means of near-surface air temperature (tas) anomalies in the pre-industrial control (piControl) scenario, an 851-year ($\nMonthsPi$-month) simulation under fixed pre-industrial greenhouse gas concentrations.

From these data, we obtain the treatment basis $\Vz$ by subtracting the long-run mean at each gridcell and computing spatial principal components by singular value decomposition of the resulting year-by-gridcell anomaly matrix. The confounder basis $\Vx$ is constructed analogously from precipitation data.
Projected scores for treatments and confounders are then generated from the following structural causal model:
\begin{equation}\label{eq:sim-scm}
\begin{aligned}
\mathbf{x}_t &= \boldsymbol\Phi_x \mathbf{x}_{t-1} + \boldsymbol\nu_t, & \boldsymbol\nu_t &\sim \mathcal{N}(\mathbf{0}, \boldsymbol\Sigma_x),\\
\mathbf{z}_t &= \boldsymbol\Phi_r \mathbf{z}_{t-1} + \PhiTwo \mathbf{x}_t + \boldsymbol\mu_r + \boldsymbol\varepsilon_t, & \boldsymbol\varepsilon_t &\sim \mathcal{N}(\mathbf{0}, \boldsymbol\Sigma_{\varepsilon,r}),\\
Y_t(s) &= g\!\left(Z_t(s),X_t(s)\right) + \epsilon_t(s), & \epsilon_t(s) &\sim \mathcal{N}(0, \sigma_Y^2),
\end{aligned}
\end{equation}
where the full treatment and confounder fields are reconstructed from the scores using the GCM-derived basis representations, $\mathbf Z_t=\Vz\mathbf z_t$ and $\mathbf X_t=\Vx\mathbf x_t$. Time series for each variable are obtained from Equation~\eqref{eq:sim-scm} by initializing at $\mathbf{x}_0=\mathbf{z}_0=\mathbf{0}$ and iterating forward in time.
For each experiment, each replicate is obtained by simulating \simSteps{} successive time points and discarding the first \simBurnIn{} to retain $T=\simT$ synthetic time steps.

\subsection{Estimator Validation and Comparison}\label{sec:sim-comparison}

To evaluate our $\kappa$-regularized estimator under the penalty-selection rule of Equation~\eqref{eq:mse-selection}, we compare it with four alternative methods: ordinary IPW, component-truncated IPW retaining the first five components, and two weight-capping methods (a square-root and an ESS-matched weight cap). For this experiment, we assume that $f_0$ and $f_1$ are known and that $f_0=f_\mathrm{obs}$. We construct $f_0$ and $f_1$ so that both their means and their variances differ across components; the parameter values and descriptions of each comparison method are given in Supplementary Section~\ref{app:sim-design}. At each time step, we reconstruct $Z_t(s)$ and generate the outcome at each location from
\begin{equation}\label{eq:sim-outcome}
 Y_t(s)=\gamma_{\mathrm{tas}}Z_t(s)+\epsilon_t(s),
 \qquad \gamma_{\mathrm{tas}}=5\times10^{-4},
 \qquad \epsilon_t(s)\sim\mathcal N\left(0,(5\times10^{-4})^2\right),
\end{equation}
where $\epsilon_t(s)$ are independent across space and time.

For each of \simExpOneRecords{} independent replicates, we compute the $\kappa$-regularized estimator and the four comparison estimators, and evaluate each against the known ground truth $\tau(s)=\gamma_{\mathrm{tas}}[\Vz\boldsymbol\delta](s)$, where $\boldsymbol\delta=\boldsymbol\mu_1-\boldsymbol\mu_0$ is the conditional mean shift between the two intervention distributions. The penalty $\kappa$ is selected from a grid whose lower boundary $\kappa=0$ corresponds to ordinary IPW; in this study, $\kappa=0$ is selected in \expOneZeroSelPct\% of replicates.

\begin{table}[tbp]
\centering
\fontsize{10}{12}\selectfont
\setlength{\tabcolsep}{3pt}
\caption{Estimator comparison in the simulation study of Section~\ref{sec:sim-comparison}, using \simExpOneRecords{} independently simulated replicates. Bias, sampling standard deviation, global RMSE, and mean spatial error are in units of $10^{-6}$. }
\label{tab:sim-comparison}
\begin{tabular*}{\textwidth}{@{\extracolsep{\fill}}lrrrrr@{}}
\hline
Method & Bias & Sampling SD & Global RMSE & Spatial error & Mean ESS \\
\hline
Ordinary IPW & 0.11 & 5.32 & 5.32 & 23.46 & 300.7 \\
Truncated to five components & \textminus{}0.52 & 4.94 & 4.97 & 23.45 & 324.0 \\
Square-root weight cap & 0.27 & 4.86 & 4.87 & 22.92 & 302.8 \\
ESS-matched weight cap & 3.20 & 3.02 & 4.40 & 18.67 & 352.0 \\
\textbf{{\boldmath $\kappa$-regularized}} & \textbf{1.50} & \textbf{2.94} & \textbf{3.30} & \textbf{18.53} & \textbf{352.0} \\
\hline
\end{tabular*}

\end{table}

The bias, sampling standard deviation, RMSE of the global estimate, mean spatial error, and weight ESS for the five estimators are reported in Table~\ref{tab:sim-comparison}. Here, spatial error is defined as the mean across replicates of the area-weighted root mean squared field error within each replicate, using cosine-latitude weights. The $\kappa$-regularized estimator gives the lowest global RMSE and mean spatial error of the five estimators. Relative to ordinary IPW, the $\kappa$-regularized estimator reduces the sampling standard deviation by over a third while incurring only a moderate bias. The $\kappa$-regularized estimator also has the highest mean effective sample size; the ESS-matched weight cap shares this ESS value by design, but with notably worse bias and worse global RMSE.

\subsection{Confounding Adjustment}\label{sec:sim-confounding}

To evaluate the ability of the $\kappa$-regularized estimator to adjust for known confounding effects at different strengths, we set the outcome in the synthetic data model of Equation~\eqref{eq:sim-scm} to be
\[
Y_t(s)=\gamma_{\mathrm{tas}}Z_t(s)+\gamma_{\mathrm{pr}}X_t(s)+\epsilon_t(s),
\]
and vary the level of the known confounding parameter $\gamma_{\mathrm{pr}}$ from zero (representing no confounding) to $\gamma_{\max}=\simExpTwoGammaMax$. Specifics on the construction of $f_1$ and $f_0$ for this experiment, and the grid of values of $\gamma_{\mathrm{pr}}$, are detailed in Supplementary Section~\ref{app:sim-confounding}.

\begin{figure}[tbp]
\centering
\includegraphics[width=\textwidth]{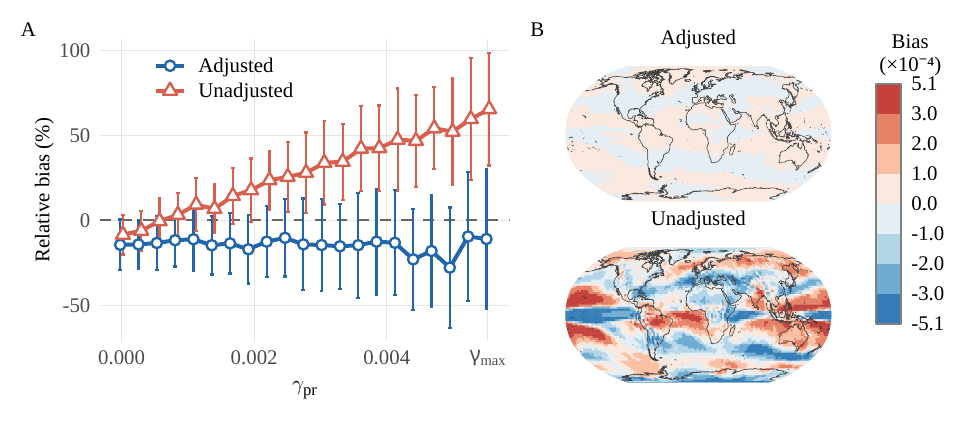}
\caption{Results of the confounding-adjustment experiment. Panel~A shows the relative bias of the adjusted and unadjusted estimators over a grid of confounding strengths $\gamma_{\mathrm{pr}}$; points show the mean over the \simExpTwoRecords{} replicates and error bars show $\pm1$ standard deviation across replicates. Panel~B shows maps of the bias of the adjusted and unadjusted estimators at $\gamma_{\mathrm{pr}}=\gamma_{\max}$, in the units of the simulated outcome.}
\label{fig:confounding}
\end{figure}

For each value of $\gamma_{\mathrm{pr}}$ in the grid, we generate \simExpTwoRecords{} independent replicates, each with its own simulated confounder, treatment, and outcome. On each replicate, we compare the $\kappa$-regularized estimator to an \textit{unadjusted} estimator whose propensity model omits the confounder. The results are displayed in Figure~\ref{fig:confounding}, where Panel~A shows the mean and spread across replicates of the relative bias of the global estimate, and Panel~B shows the bias at each gridcell for $\gamma_{\mathrm{pr}}=\gamma_{\max}$. At $\gamma_{\mathrm{pr}}=0$ (no confounding), the replicate intervals of both estimators contain zero. As $\gamma_{\mathrm{pr}}$ increases, the unadjusted relative bias increases steadily and reaches $\simExpTwoBiasMargHigh\%$ at $\gamma_{\mathrm{pr}}=\gamma_{\max}$. The replicate intervals of the adjusted estimator, in contrast, contain zero for all values of $\gamma_{\mathrm{pr}}$, although their spread increases with the strength of confounding.

Panel~B shows the spatial distribution of bias for both estimators at $\gamma_{\mathrm{pr}}=\gamma_{\max}$. The bias of the adjusted estimator is smaller in magnitude than that of the unadjusted estimator at nearly every gridcell. The bias of the unadjusted estimator is far from spatially uniform, with large negative bias over the equatorial Pacific, Antarctica, and the mid-latitude oceans of both hemispheres, most strongly near $45^\circ$S, and positive bias over Australia, the Amazon, the subtropical oceans, and the Southern Ocean near $60^\circ$S. 

\section{Application}\label{sec:application}

We apply the proposed framework to the estimation of the direct effect of anthropogenic temperature change on vapor-pressure deficit, adjusted for confounding by precipitation.
To focus on estimation of the forced component of the causal contrast, we take as observations the GISS-E2-1-G piControl simulation described in Section~\ref{sec:sim-data};
the extension of the framework to observational data from a reanalysis product is described in Section~\ref{sec:discussion} and reserved for future work.

\subsection{Data}\label{sec:app-data}

We estimate the forced intervention $f_1$ using runs from the \textit{historical} scenario of GISS-E2-1-G; this ensemble consists of \nHistMembers{} members driven by historical levels of greenhouse gas, aerosol, and natural forcing \citep{kelley2020giss, nasagiss_cmip6}. We use the 1950--2014 window (780 months per member), where 2014 is the end year of the CMIP6 historical simulations. To extend to the near present (taken to be 2025), we augment the historical forcing scenarios with projections from the intermediate Shared Socioeconomic Pathway SSP2-4.5 \citep{nasagiss_cmip6}; details on the variables used, the VPD formula, historical ensemble members, and the SSP extension are described in Supplementary Sections~\ref{app:variables}--\ref{app:ssp-pathway}. Representative samples from the deseasonalized piControl fields, alongside the ensemble-mean historical-minus-piControl temperature difference, are shown in Figure~\ref{fig:data-overview}. While the forced temperature response is concentrated at high latitudes (Panel~A), VPD anomalies (Panel~D) are largest in magnitude at low latitudes, because saturation vapor pressure, and hence VPD and its variability, is much lower in cold regions.

Since anthropogenic influence is the dominant forcing in the historical simulations, the causal contrast between $f_1$ and the pre-industrial control $f_0$ largely represents the direct effect of human-driven temperature increases on vapor-pressure deficit, adjusted for precipitation. The historical simulations also include natural forcings, such as volcanic aerosols; a more targeted analysis of the effect of greenhouse gas concentrations in isolation is possible with additional scenarios and is discussed in Section~\ref{sec:discussion}.

\begin{figure}[tbp]
\centering
\includegraphics[width=\textwidth]{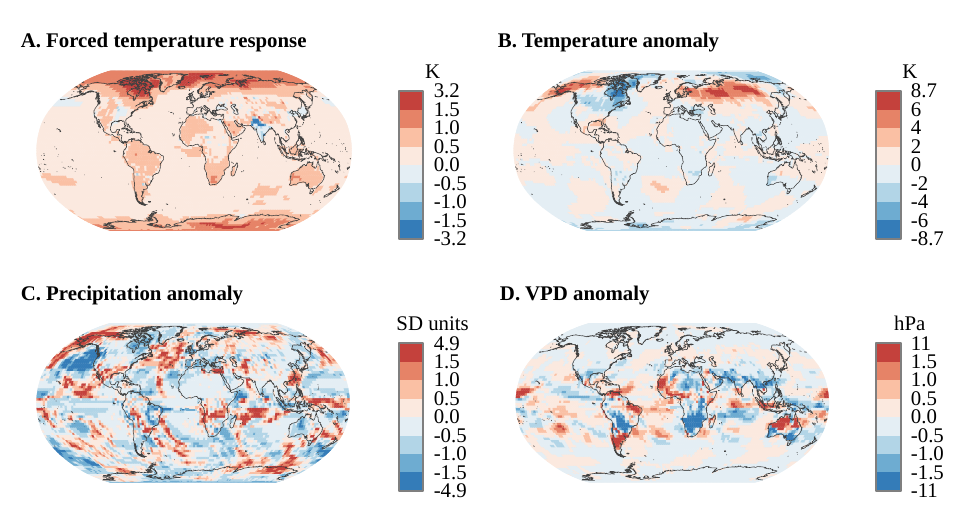}
\caption{Visualizations of the GISS-E2-1-G data used in the application. (A) the ensemble-mean historical-minus-piControl near-surface air temperature, representing the forced temperature response; (B)--(D) monthly spatial fields from the piControl simulation at month 600 (piControl year 50): (B) near-surface air temperature anomaly, (C) precipitation anomaly (in standard-deviation units), and (D) VPD anomaly, all displayed on the common $2.5^\circ$ analysis grid.}
\label{fig:data-overview}
\end{figure}

\subsection{Intervention Estimation and Validation}\label{sec:app-estimation}

To estimate $f_0$ and $f_1$, we modify the construction of Section~\ref{sec:intervention} through a hybrid basis representation and through the addition of seasonally varying intercepts and residual covariances.
The hybrid basis $\Vz$ combines the two leading forced-response patterns (the leading principal components of the annual ensemble-mean historical anomalies over 1950--2014) with the \KdimControl{} leading piControl principal components of the deseasonalized temperature anomaly field, orthogonalized against the two forced-response patterns, resulting in $\Kz=\Kdim$.
This hybrid basis construction allows the intervention to capture the leading components of the forced response while retaining much of the pre-industrial variability structure; a validation study justifying this construction is reported in Supplementary Section~\ref{app:eof-stability}.

The control intervention $f_0$ is estimated from the piControl record using the following form of Equation~\eqref{eq:intervention-dist}, in which the intercepts and residual covariances vary by calendar month:
\begin{equation*}
 f_0(\mathbf{z}_t \mid \mathbf{z}_{t-1}, \mathbf{x}_t)
 = \mathcal{N}\!\left(\mathbf{z}_t;\;
   \boldsymbol\Phi_0\mathbf{z}_{t-1}
   + \PhiTwo\mathbf{x}_t + \boldsymbol\mu_{0,m(t-1)},\;
   \boldsymbol\Sigma_{\varepsilon,0,m(t-1)}\right),
\end{equation*}
where $m(t)$ denotes the calendar month of time step $t$.
Following Section~\ref{sec:one-step}, this density does not depend on $t^{*}$, since the control distribution is assumed to be stationary, which is appropriate in this setting because the piControl run has no forced trend. The forced intervention $f_1$ is estimated from the historical simulations with the analogous parameterization
\begin{equation*}
 f_{1;t^{*}}(\mathbf{z}_t \mid \mathbf{z}_{t-1}, \mathbf{x}_t)
 = \mathcal{N}\!\left(\mathbf{z}_t;\;
   \boldsymbol\Phi_1\mathbf{z}_{t-1}
   + \PhiTwo\mathbf{x}_t + \boldsymbol\mu_{1,m(t-1)}(t^{*}),\;
   \boldsymbol\Sigma_{\varepsilon,1,m(t-1)}\right),
\end{equation*}
where $\boldsymbol\mu_{1,m}(t^{*})$ is the month-specific realization of the intervention intercept $\boldsymbol\mu_1(t^{*})$ of Section~\ref{sec:one-step}. The intervention-year component of the intercept, $\boldsymbol\mu_1(t^{*})$, is first estimated for each intervention year $t^*$ from 1951 to 2025 and subtracted from the treatment scores, from which $\boldsymbol\Phi_1$, the month-specific intercepts, and $\boldsymbol\Sigma_{\varepsilon,1,m}$ are then estimated (Supplementary Section~\ref{app:model-validation}). The confounder coefficient $\PhiTwo$ is shared by the two interventions and is estimated from the piControl record along with the other parameters of $f_0$. These specifications of $f_0$ and $f_1$ were chosen through an out-of-sample cross-validation study comparing models with shared versus seasonally varying transition matrices, intercepts, and residual covariances, the full details of which are provided in Supplementary Section~\ref{app:model-validation}. 

To account for seasonality in the treatment effect, we modify the $\kappa$-regularized estimator of Equation~\eqref{eq:kappa-regularized} to normalize the weights within each calendar month, and then average the twelve monthly contrasts with equal weight. Formally, let $\mathcal J_m$ denote the transitions from source calendar month $m$, and let $n_m=|\mathcal J_m|$; the monthly averaged estimator is then
\begin{equation*}\label{eq:hajek-monthly}
 \widehat\tau_{t^{*}}(s;\kappa)=\frac1{12}\sum_{m=1}^{12}\left\{
 \frac{\sum_{t\in\mathcal J_m}w_t^{(1;t^*)}(\kappa)Y_t(s)}{\sum_{t\in\mathcal J_m}w_t^{(1;t^*)}(\kappa)}
 -\frac1{n_m}\sum_{t\in\mathcal J_m}Y_t(s)\right\}.
\end{equation*}
The second term of the contrast takes the form of the unweighted sample mean because the observations are from the piControl simulation, in which case $f_0=f_{\mathrm{obs}}$ and the control weights equal one; in future extensions to reanalysis data as observations this term will also be weighted. 

\subsection{Results}\label{sec:app-results}

We evaluate the monthly averaged $\kappa$-regularized estimator at annual intervention times $t^{*}$ from 1951 to 2025.
The regularization parameter $\kappa$ is selected using Equation~\eqref{eq:mse-selection}, with $a_s$ the normalized cosine-latitude area weights of the gridcells and $\mathcal{T}$ the final five intervention years, 2021--2025.
The selection rule yields $\widehat\kappa=\kappaPresent$, and at this penalty the effective sample size of the weights is \essFirstJoint{} at $t^{*}=\essFirstYear$ and \essPresentJoint{} at $t^{*}=2025$, out of \nMonthsPi{} piControl months (the full ESS time series is displayed in Supplementary Section~\ref{app:ess}); note that this decrease in effective sample size with $t^*$ is expected as overlap decreases with forcing over time. To visualize the effect of the weighting, months with high and low density under the two fitted interventions are shown in Supplementary Section~\ref{app:weight-diagnostics}.

The resulting spatially averaged estimates for each intervention year $t^*$, along with 95\% intervals computed using the block bootstrap of Section~\ref{sec:bootstrap}, are displayed in Figure~\ref{fig:intervention-year}. The first year from which the global interval excludes zero at every later intervention time is \emergenceYear{}, marked with a vertical dotted line. Several visible short-term decreases in the estimated effect are consistent with volcanic eruptions included in the historical forcing scenario, as the stratospheric aerosols from a large eruption result in surface cooling, lowering the saturation vapor pressure and hence VPD \citep{robock2000volcanic, soden2002global}. Figure~\ref{fig:intervention-year} marks the years following the three major eruptions in the study period: Agung in 1963, El Chich\'{o}n in 1982, and Pinatubo in 1991 \citep{gvp2025votw}. 

In addition to the global mean, Figure~\ref{fig:intervention-year} displays the estimates averaged over land and ocean masks, with cosine-latitude weights within each region (details are given in Supplementary Section~\ref{app:land-sea-mask}). The causal effect of temperature on VPD is notably stronger over land than globally, whereas the ocean-mean estimate remains close to zero across all intervention years. This land--ocean contrast is expected because, under warming, relative humidity declines over land but remains nearly constant over the ocean \citep{byrne2018trends}.

\begin{figure}[tbp]
\centering
\includegraphics[width=0.92\textwidth]{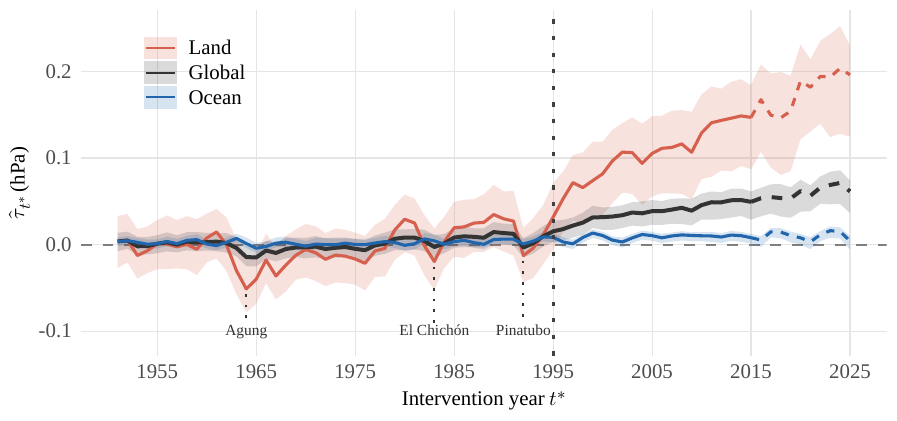}
\caption{Estimates of the causal effect of temperature on vapor-pressure deficit (hPa) averaged globally and over land and ocean masks as a function of intervention year $t^{*}$, along with 95\% block-bootstrap percentile intervals. The vertical line marks \emergenceYear{}, from which every global interval excludes zero. Intervention years after 2014 use the SSP2-4.5 extension of the historical trajectory and are drawn with dashed lines. The years following three major volcanic eruptions are annotated.}
\label{fig:intervention-year}
\end{figure}

Table~\ref{tab:application-regional} gives the global, land, and ocean estimates with 95\% intervals for the final intervention year, $t^{*}=2025$. The global and land-mask intervals exclude zero, whereas the ocean-mask interval includes zero. To put the scale of these estimates into context with respect to natural variability, the piControl global-mean VPD is \vpdMeanGlobal~hPa and the standard deviation of its annual means is \vpdSdAnnualGlobal~hPa; the global effect is therefore \vpdGlobalPresentSdAnnual{} standard deviation of the annual means, and the land effect is \vpdLandPresentSdAnnual{} standard deviations of the corresponding land-mean series.

A sense of the scale of these estimates can also be obtained through comparison to the expected response given thermodynamics alone. Vapor-pressure deficit is the difference between saturation vapor pressure, which by the Clausius--Clapeyron relation rises by about 7\% per degree of warming, and actual vapor pressure, which depends on the water content of the air (Supplementary Section~\ref{app:variables}). When the estimator's weights are applied to saturation vapor pressure, which can be computed from temperature alone, the Clausius--Clapeyron relation predicts a land response of \ccFixedLand{}~hPa when \emph{specific humidity} is held fixed and \ccConstRhLand{}~hPa when \emph{relative humidity} is held fixed. The land estimate of \vpdLandPresent{}~hPa lies between these two values, indicating that while the water content of the air over land rises with temperature, it does not rise enough to maintain relative humidity. The ocean estimate of \vpdOceanPresent{}~hPa lies below its constant-relative-humidity reference of \ccConstRhOcean{}~hPa, broadly consistent with the near-constant marine relative humidity under warming \citep{held2006robust}.

\begin{table}[tbp]
\centering
\caption{Estimates of the causal effect of temperature on vapor-pressure deficit (hPa) at $t^{*}=2025$, averaged globally and over the land and ocean masks, for the precipitation-adjusted and unadjusted estimators as well as their paired difference (adjusted minus unadjusted), along with 95\% block-bootstrap percentile intervals.}
\label{tab:application-regional}
\begin{tabular*}{\textwidth}{@{\extracolsep{\fill}}lrrrrrr@{}}
\hline
 & \multicolumn{2}{c}{Adjusted} & \multicolumn{2}{c}{Unadjusted} & \multicolumn{2}{c@{}}{Difference} \\
Region & Estimate & 95\% interval & Estimate & 95\% interval & Estimate & 95\% interval \\
\hline
Global & $0.061$ & $[0.036,\,0.074]$ & $0.090$ & $[0.064,\,0.103]$ & $-0.029$ & $[-0.052,\,-0.010]$ \\
Land & $0.196$ & $[0.125,\,0.230]$ & $0.275$ & $[0.195,\,0.313]$ & $-0.079$ & $[-0.139,\,-0.021]$ \\
Ocean & $0.004$ & $[-0.006,\,0.010]$ & $0.012$ & $[0.007,\,0.018]$ & $-0.008$ & $[-0.016,\,-0.002]$ \\
\hline
\end{tabular*}

\end{table}

Figure~\ref{fig:application-maps} shows maps of the gridcell point estimates, their bootstrap standard errors, and their signal-to-noise ratios at $t^*=2025$.
Gridcells with an absolute signal-to-noise ratio greater than two cover \snrAreaPct\% of the globe: \snrAreaPosPct\% with positive estimates and \snrAreaNegPct\% with negative estimates.  The land--ocean contrast is particularly apparent here; the estimated effects are largest in absolute value over continental equatorial regions and are generally smaller over the oceans and at high latitudes. The small effects at high latitudes can be contrasted with the forced temperature response, which is largest at high latitudes (Figure~\ref{fig:data-overview}A); the difference can be explained by the fact that VPD values are generally smaller in cold regions because saturation vapor pressure is lower at low temperatures (Figure~\ref{fig:data-overview}D). The standard errors are also much larger over continental equatorial regions (Figure~\ref{fig:application-maps}B) than at high latitudes and over the ocean.

\begin{figure}[tbp]
\centering
\includegraphics[width=\textwidth]{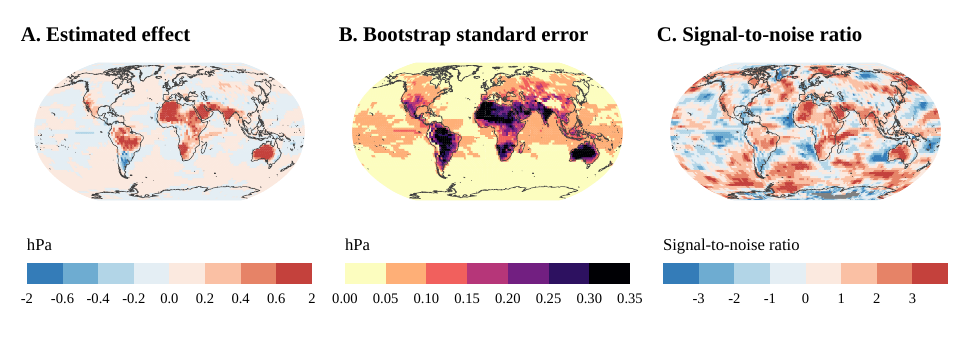}
\caption{Vapor-pressure-deficit estimates at $t^{*}=2025$. Panel~A shows point estimates (hPa), Panel~B bootstrap standard errors (hPa), and Panel~C signal-to-noise ratios, defined as the point estimate divided by the bootstrap standard error.}
\label{fig:application-maps}
\end{figure}

To assess the effect that confounding adjustment has on the estimates, we additionally apply to the data an unadjusted estimator, which uses the same weighting construction but omits precipitation from the densities $f_0$ and $f_1$. The unadjusted estimates at $t^*=2025$ are shown alongside the adjusted estimates in Table~\ref{tab:application-regional}, together with the paired bootstrap differences between the two estimators. In each region, the unadjusted estimate exceeds the adjusted version and the interval for the difference excludes zero. The reduction of the estimates under adjustment is expected given the causal structure of Figure~\ref{fig:dag}; since precipitation lowers temperature through evaporative cooling and increased cloud cover, and lowers VPD through increased near-surface humidity, low precipitation anomalies produce both higher temperatures and higher VPD. High temperatures therefore co-occur with low precipitation in the observational record, and without adjustment the association between temperature and VPD overstates the causal effect of temperature alone.  The effect of adjustment over space is displayed in Figure~\ref{fig:adjustment-comparison}A, where the generally larger effect of adjustment over land can be seen.

\begin{figure}[tbp]
\centering
\includegraphics[width=\textwidth]{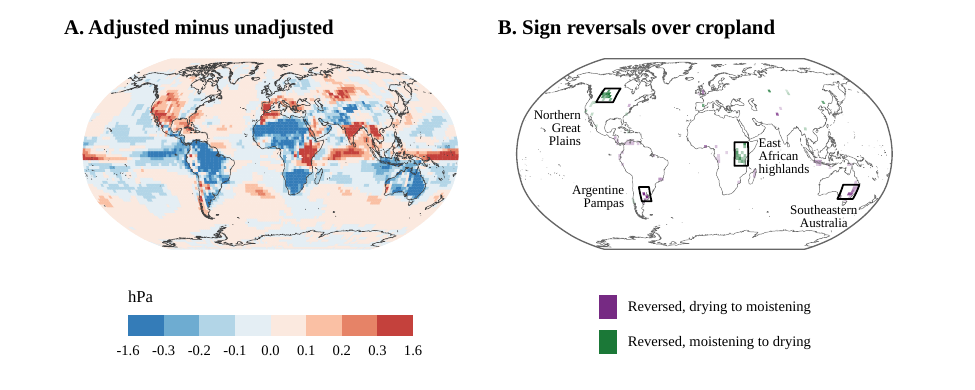}
\caption{Confounding-adjusted and unadjusted gridcell estimates of the VPD effect at $t^{*}=2025$. Panel~A shows the adjusted minus the unadjusted estimate (hPa). Panel~B shows cropland-masked gridcells on which adjustment reverses the sign of the estimate, colored by the direction of reversal, with darker shading indicating a larger cropland fraction. Boxes outline four agricultural regions that contain most of the sign-reversed cropland.}
\label{fig:adjustment-comparison}
\end{figure}

\subsection{Implications for Agricultural Adaptation}\label{sec:app-adaptation}

Because VPD is a principal driver of crop stress \citep{lobell2014greater}, the consequences of adjustment for adaptation matter most over cropland. To measure how much agricultural area is affected, we take the cropland fraction of each gridcell from the Land-Use Harmonization 2 (LUH2) land-use states for \cropYear{}, the land-use forcing of the historical simulations \citep{hurtt2020harmonization}, and weight each gridcell by its cropland area. Globally, adjustment reverses the sign of the estimate over \adjClassCropRevPct\% of cropland, amounting to \adjClassCropRevMha{} million hectares.

As the effect of adjustment is strongly regional, Figure~\ref{fig:adjustment-comparison}B highlights four agricultural regions that together hold \cropWindowsRevPct\% of the cropland where adjustment reverses the sign; enlarged views of these regions are shown in Supplementary Figure~\ref{fig:cropland-zooms}. Adjustment changes the sign of the estimate from negative to positive in the Northern Great Plains and the East African highlands (from a moistening to a drying effect), whereas over parts of the Argentine Pampas and southeastern Australia it changes the sign from positive to negative (from drying to moistening). Where the unadjusted estimate has the wrong sign, an analysis that omits precipitation may misdirect adaptation. In the Pampas and southeastern Australia, it may favor heat-focused measures, such as heat-tolerant crop varieties, when moisture-focused measures, such as irrigation and water storage, would be more effective; in the Northern Great Plains and the East African highlands, it may discount heat-focused measures where they are needed. Supplementary Section~\ref{app:cropland} gives the cropland mask, and Supplementary Table~\ref{tab:crop-regions} gives the cropland of these four regions and of \cropOtherN{} further major cropland regions by adjustment class, that is, by whether adjustment strengthens, weakens, or reverses the unadjusted estimate, or whether both estimates are negligible.

\section{Discussion and Future Work}\label{sec:discussion}

This work develops a spatio-temporal stochastic-intervention framework and regularized estimator for observational causal inference on global climate fields, and applies this framework to the estimation of the causal effect of historical forced temperature change on aridity measured by vapor-pressure deficit.
The framework is supported by an identification result (Supplementary Proposition~\ref{prop:identification}), by two simulation studies with known ground truth (Section~\ref{sec:simulation}), and by an out-of-sample validation of the intervention models (Supplementary Section~\ref{app:model-validation}). The application demonstrates the value of the framework through comparison with an estimator that does not adjust for precipitation: adjustment lowers the global, land, and ocean estimates and reverses the sign of the estimate over \adjClassCropRevPct\% of global cropland, amounting to \adjClassCropRevMha{} million hectares. As the sign of the direct effect determines whether heat-focused or moisture-focused adaptation is the relevant response, these results have the potential to improve adaptation responses in several major crop-growing regions.

Given the complexity of causal inference in climate science, the proposed approach has several methodological and scientific limitations.
The first methodological limitation is that the estimator, being based on inverse probability weighting, requires a correctly specified propensity model \citep{kang2007demystifying}. 
Future work will relax this requirement through an extension to doubly robust estimators \citep{diaz2021nonparametric}, which introduce an outcome model and require correct specification of only one of the outcome and propensity models. 
This extension is non-trivial, however, since overlap between the observational record and both the control and forced interventions must then also be addressed in the outcome model. 

Second, the framework assumes that the potential outcomes depend on the treatment only through the projected scores, that is, $Y_t(s;\mathbf{z}_t)=Y_t(s;\mathbf{Z}_t)$. While retaining a relatively large number of projected components makes this assumption more plausible, the degree to which the full-dimensional causal estimand is represented in the projected space remains an open question for further investigation.

Third, the block bootstrap resamples only the observational record, so its intervals omit uncertainty in the intervention parameters estimated from the historical ensemble, in the estimation of the penalty parameter $\kappa$, and in the estimated basis representations $\Vz$ and $\Vx$.
Future work will improve the standard errors by resampling climate-model ensemble members in addition to temporal blocks, and by investigating analytical alternatives to the bootstrap such as sandwich estimators.

On the scientific side, the first limitation is that the causal structure of Figure~\ref{fig:dag}, while physically motivated, is a simplified representation of the underlying dynamics.
In particular, it omits confounders such as soil moisture, cloud cover, and radiation; it excludes mediators and lagged feedback from temperature to precipitation; and it uses a single monthly lag, so that longer-memory causal effects and feedback processes are not represented.
Richer causal structures, including the incorporation of additional confounders, the presence of mediating effects, and the representation of feedback effects over time, will be considered in future work.

A second scientific limitation concerns the use of only historical simulations for estimating the forced intervention, and the use of a single climate model for both interventions.
Because the historical scenario includes natural as well as anthropogenic forcing, these results do not isolate the effects of increasing greenhouse gas concentrations or anthropogenic influence more generally, as can be seen from the appearance of volcanic effects in Figure~\ref{fig:intervention-year}.
These components can be isolated through the use of more targeted scenarios, such as those of the Detection and Attribution Model Intercomparison Project \citep{gillett2016detection}.
Additional climate models for both control and forced scenarios would permit an extension of the bootstrap standard errors and the quantification of climate-model uncertainty through inter-model spread.

Finally, our application takes a climate-model control simulation as the observational record, both to isolate the forced component of the causal contrast and to set aside the additional sources of uncertainty inherent to real observations.
Using reanalysis products as observations would introduce uncertainty from spatial interpolation and from varying temporal record lengths, which would be straightforward to incorporate in the framework. It would also require estimating $f_{\mathrm{obs}}$ separately from $f_0$, which poses additional challenges because the observational record is much shorter than the climate-model simulation and because $f_{\mathrm{obs}}$ is non-stationary over the historical record.
Future work will explore approaches applying this framework to reanalysis data while maintaining detection power.

\begin{acks}[Acknowledgments]
We acknowledge the World Climate Research Programme, which, through its Working Group on Coupled Modelling, coordinated and promoted CMIP6. We thank the climate modeling groups, in particular the NASA Goddard Institute for Space Studies for the GISS-E2-1-G simulations, for producing and making available their model output, the Earth System Grid Federation (ESGF) for archiving the data and providing access, and the multiple funding agencies who support CMIP6 and ESGF.
\end{acks}

\begin{funding}
The author received no specific funding for this work.
\end{funding}

\begin{acks}[Data availability]
The GISS-E2-1-G simulations used in this work are publicly available through the Earth System Grid Federation \citep{nasagiss_cmip6}: the piControl run, the historical ensemble, and the SSP1-2.6, SSP2-4.5, and SSP5-8.5 projections. The LUH2 land-use states are available through input4MIPs \citep{hurtt2017luh2}. The sources of the greenhouse gas concentrations and of the land--sea mask are given in the Supplementary Material, and the dataset identifier and version of every input file are listed with the code.
\end{acks}

\bigskip
\noindent\textbf{Code availability.}
The \texttt{stsi} R package is available at \url{https://github.com/samjbaugh/stsi}. The code that reproduces every analysis, figure, table, and reported number, with the results and the list of input files, is available at \url{https://github.com/samjbaugh/stsi_analysis}. The Supplementary Material follows the references in this document.

\bibliographystyle{apalike}  
\bibliography{references}

\clearpage
\setcounter{section}{0}
\renewcommand{\theHsection}{supp.\arabic{section}}
\setcounter{subsection}{0}
\renewcommand{\theHsubsection}{supp.\arabic{subsection}}
\setcounter{table}{0}
\renewcommand{\theHtable}{supp.\arabic{table}}
\setcounter{figure}{0}
\renewcommand{\theHfigure}{supp.\arabic{figure}}
\setcounter{equation}{0}
\renewcommand{\theHequation}{supp.\arabic{equation}}
\setcounter{theorem}{0}
\renewcommand{\theHtheorem}{supp.\arabic{theorem}}
\setcounter{lemma}{0}
\renewcommand{\theHlemma}{supp.\arabic{lemma}}
\setcounter{corollary}{0}
\renewcommand{\theHcorollary}{supp.\arabic{corollary}}
\setcounter{proposition}{0}
\renewcommand{\theHproposition}{supp.\arabic{proposition}}
\setcounter{definition}{0}
\renewcommand{\theHdefinition}{supp.\arabic{definition}}
\setcounter{assumption}{0}
\renewcommand{\theHassumption}{supp.\arabic{assumption}}
\setcounter{example}{0}
\renewcommand{\theHexample}{supp.\arabic{example}}
\setcounter{remark}{0}
\renewcommand{\theHremark}{supp.\arabic{remark}}
\renewcommand{\theHsubsection}{\theHsection.\arabic{subsection}}
\renewcommand{\thesection}{S\arabic{section}}
\renewcommand{\thesubsection}{S\arabic{section}.\arabic{subsection}}
\renewcommand{\thetable}{S\arabic{table}}
\renewcommand{\thefigure}{S\arabic{figure}}
\renewcommand{\theequation}{S\arabic{equation}}
\renewcommand{\thetheorem}{S\arabic{theorem}}
\renewcommand{\thelemma}{S\arabic{lemma}}
\renewcommand{\thecorollary}{S\arabic{corollary}}
\renewcommand{\theproposition}{S\arabic{proposition}}
\renewcommand{\thedefinition}{S\arabic{definition}}
\renewcommand{\theassumption}{S\arabic{assumption}}
\renewcommand{\theexample}{S\arabic{example}}
\renewcommand{\theremark}{S\arabic{remark}}
\section*{Supplementary Material}
\addcontentsline{toc}{section}{Supplementary Material}
The Supplementary Material contains the proof of the identification result (Section~\ref{app:identification}); the parameter values, penalty grid, comparison methods, and confounding-experiment design of the simulation studies (Section~\ref{app:simulation}); and supplementary details and diagnostics for the application (Section~\ref{app:diagnostics}): the variables and ensembles; the scenario extension to 2025; the application basis and its held-out comparison with principal-component bases; the out-of-sample validation of the propensity and intervention models; the penalty grid, the effective sample size by intervention year, and the conditioning of the residual covariances; the land and ocean masks; the cropland mask and named agricultural regions; and representative weighted months. Section, equation, figure, and assumption references without an S prefix refer to the main text.

\section{Identification}\label{app:identification}

This section proves the identification result of Section~\ref{sec:assumptions}, Equation~\eqref{eq:identification}. The argument follows results for stochastic interventions at a single time point \citep{diaz2012population, kennedy2019nonparametric}, applied at one time step conditional on the preceding treatment and the current confounder; \citet{papadogeorgou2022causal} give the corresponding result for stochastic interventions on spatio-temporal treatments.

As in Section~\ref{sec:intervention}, write $\mathbf{c}_t = (\mathbf{z}_{t-1}, \mathbf{x}_t)$ for the conditioning pair, and let $\mathbb{E}_{\mathrm{obs}}$ denote expectation under the observational process (written $\mathbb{E}_{f_{\mathrm{obs}}}$ in Equation~\eqref{eq:identification}), in which $\mathbf{z}_t$ is drawn from $f_{\mathrm{obs}}(\cdot \mid \mathbf{c}_t)$. The intervention is applied only at time $t$: the history through ${t-1}$ and the confounder $\mathbf{x}_t$ retain their observational joint distribution, and $\mathbf{z}_t$ is then drawn from $f_{r;t^{*}}(\cdot \mid \mathbf{c}_t)$ by external randomization, independently of the potential outcomes given $\mathbf{c}_t$. The potential outcomes are the same under observation and under the intervention; by consistency, the realized outcome is $Y_t(s;\mathbf{z}_t)$. The common history distribution is part of this intervention definition. Assumption~\ref{asm:auxiliary} specifies the common conditional confounder and outcome mechanisms.

All treatment densities are with respect to Lebesgue measure on $\mathbb{R}^{\Kz}$. Write $w_t^{(r;t^*)}=f_{r;t^*}(\mathbf{z}_t\mid\mathbf{c}_t)/f_{\mathrm{obs}}(\mathbf{z}_t\mid\mathbf{c}_t)$, assigning it zero wherever the denominator is zero. In addition to Assumptions~\ref{asm:consist}--\ref{asm:auxiliary}, require, for each $r$, $t^*$, $t$, and $s$,
\begin{equation}\label{eq:identification-integrability}
\mathbb{E}_{\mathrm{obs}}\!\left[w_t^{(r;t^*)}\,|Y_t(s)|\right]<\infty.
\end{equation}
This condition ensures a finite absolute outcome mean under each intervention. Conditional expectations and conditional density identities below are understood almost everywhere under the corresponding conditioning distribution.

\begin{proposition}[Identification]\label{prop:identification}
For the intervention defined above, suppose that Assumptions~\ref{asm:consist}--\ref{asm:auxiliary} and condition~\eqref{eq:identification-integrability} hold, and that the positivity condition of Assumption~\ref{asm:pos} holds for every $\mathbf{c}_t$ in the support of its observational distribution. Then, for each $s$, $t$, $t^{*}$, and $r \in \{0, 1\}$, and almost every $\mathbf{c}_t$ under the observational distribution,
\begin{equation}\label{eq:identification-conditional}
\mathbb{E}_{\mathbf{z} \sim f_{r;t^{*}}(\cdot \mid \mathbf{c}_t)}\!\left[Y_t(s;\mathbf{z}) \mid \mathbf{c}_t\right] = \mathbb{E}_{\mathrm{obs}}\!\left[\frac{f_{r;t^{*}}(\mathbf{z}_t \mid \mathbf{c}_t)}{f_{\mathrm{obs}}(\mathbf{z}_t \mid \mathbf{c}_t)}\, Y_t(s) \;\middle|\; \mathbf{c}_t\right].
\end{equation}
If, in addition, the observational process is stationary, so that $\mathbf{c}_t \sim \pi_{x,z}$ for every $t$, then the positivity requirement above reduces to Assumption~\ref{asm:pos}, averaging the left-hand side over $\pi_{x,z}$ gives the mean potential outcome $\psi_{r;t^{*}}(s)$ of Equation~\eqref{eq:psi-def}, and averaging the right-hand side gives the inverse-probability-weighting expression~\eqref{eq:identification}. The estimand~\eqref{eq:one-step-averaged} is therefore identified by the difference between expression~\eqref{eq:identification} at $r=1$ and at $r=0$.
\end{proposition}

\begin{proof}
\noindent\emph{Step 1 (change of the joint distribution).} Let $P_{\mathrm{obs}}$ and $P_{r;t^*}$ denote the joint distributions of the observed history through ${t-1}$ and $(\mathbf{x}_t,\mathbf{z}_t,\mathbf{Y}_t)$ under observation and under the intervention, respectively. The history distribution is common by construction. In the conditional factorization of Assumption~\ref{asm:auxiliary}, the factors $f_X$ and $f_Y$ are also common, so only the treatment factor changes. Positivity (Assumption~\ref{asm:pos}) therefore gives $P_{r;t^*}\ll P_{\mathrm{obs}}$ and
\begin{equation*}
\frac{dP_{r;t^*}}{dP_{\mathrm{obs}}}
=\frac{f_{r;t^*}(\mathbf{z}_t\mid\mathbf{c}_t)}{f_{\mathrm{obs}}(\mathbf{z}_t\mid\mathbf{c}_t)}
=w_t^{(r;t^*)},
\qquad P_{\mathrm{obs}}\text{-almost surely}.
\end{equation*}
Let $\mathbb{E}_{r;t^*}$ denote expectation under $P_{r;t^*}$. Consequently, by~\eqref{eq:identification-integrability},
\begin{equation*}
\mathbb{E}_{r;t^*}[|Y_t(s)|]=\mathbb{E}_{\mathrm{obs}}[w_t^{(r;t^*)}|Y_t(s)|]<\infty.
\end{equation*}

\noindent\emph{Step 2 (conditioning on $\mathbf{c}_t$).} Positivity and normalization of the intervention density imply
\begin{equation*}
\mathbb{E}_{\mathrm{obs}}\!\left[w_t^{(r;t^*)}\mid\mathbf{c}_t\right]
=\int f_{r;t^*}(\mathbf{z}\mid\mathbf{c}_t)\,d\mathbf{z}=1
\end{equation*}
almost everywhere. Under $P_{r;t^*}$ the treatment $\mathbf{z}_t$ is drawn from $f_{r;t^*}(\cdot\mid\mathbf{c}_t)$ independently of the potential outcomes given $\mathbf{c}_t$, and by consistency (Assumption~\ref{asm:consist}) the realized outcome is $Y_t(s;\mathbf{z}_t)$, so
\begin{equation*}
\mathbb{E}_{\mathbf{z} \sim f_{r;t^{*}}(\cdot \mid \mathbf{c}_t)}\!\left[Y_t(s;\mathbf{z})\mid\mathbf{c}_t\right]
=\mathbb{E}_{r;t^*}\!\left[Y_t(s)\mid\mathbf{c}_t\right].
\end{equation*}
Conditional change of measure then gives
\begin{equation*}
\mathbb{E}_{r;t^*}\!\left[Y_t(s)\mid\mathbf{c}_t\right]
=\frac{\mathbb{E}_{\mathrm{obs}}[w_t^{(r;t^*)}Y_t(s)\mid\mathbf{c}_t]}{\mathbb{E}_{\mathrm{obs}}[w_t^{(r;t^*)}\mid\mathbf{c}_t]}
=\mathbb{E}_{\mathrm{obs}}[w_t^{(r;t^*)}Y_t(s)\mid\mathbf{c}_t],
\end{equation*}
which is~\eqref{eq:identification-conditional}. 

\noindent\emph{Step 3 (averaging over $\mathbf{c}_t$).} Under stationarity the observational distribution of $\mathbf{c}_t$ is $\pi_{x,z}$ for every $t$, so averaging~\eqref{eq:identification-conditional} over $\mathbf{c}_t \sim \pi_{x,z}$ gives $\psi_{r;t^{*}}(s)$ of Equation~\eqref{eq:psi-def} on the left and, by iterated expectations, the expression~\eqref{eq:identification} on the right. Taking the difference between the identities for $r = 1$ and $r = 0$ identifies $\tau_{t^{*}}(s)$; the control term $\psi_0(s)$ does not depend on $t^{*}$ because $f_0$ does not.
\end{proof}

The conditional identity~\eqref{eq:identification-conditional} does not require stationarity; stationarity is used only in Step~3, to equate the observational distribution of the conditioning pair at each observation time with the averaging distribution $\pi_{x,z}$ of~\eqref{eq:psi-def}. The outcome factor $f_Y$ of Assumption~\ref{asm:auxiliary} may depend on the full history $\mathcal{F}_{t-1}$; this factor cancels in Step~1 and therefore does not enter the weights. When $f_0 = f_{\mathrm{obs}}$, the control density ratio is one almost surely and $\psi_0(s) = \mathbb{E}_{\mathrm{obs}}[Y_t(s)]$.

\section{Simulation Study Details}\label{app:simulation}

\subsection{Estimator Comparison Experiment}\label{app:sim-design}

This section gives the intervention distributions, the penalty grid, and the comparison methods of the estimator comparison of Section~\ref{sec:sim-comparison}.

\noindent\textbf{Intervention distributions and penalty grid.} The treatment is represented by $\Kz=20$ components, and $\lambda_k$ denotes the variance of the $k$th principal component of the annual piControl temperature anomalies (Section~\ref{sec:sim-data}); in the simulations, these eigenvalues are used as the penalty scales $\lambda_k$ of Section~\ref{sec:propensity}. Both intervention distributions have $\boldsymbol\Phi_0=\boldsymbol\Phi_1=0.3\mathbf I_{20}$ and $\PhiTwo=\mathbf0$. The control distribution, which is also the observational distribution, has $\boldsymbol\mu_0=\mathbf0$ and $\boldsymbol\Sigma_{\varepsilon,0}=\operatorname{diag}(\lambda_1,\ldots,\lambda_{20})$. The forced distribution has
\begin{gather*}
\boldsymbol\mu_1=0.7\,(-10,-8,-6,4,-3,-2,-2,0.5,2,0.1,0,\ldots,0)^\top,\\
\boldsymbol\Sigma_{\varepsilon,1}=\operatorname{diag}(\rho_1\lambda_1,\ldots,\rho_{20}\lambda_{20}),
\end{gather*}
where $(\rho_1,\ldots,\rho_{10})=(1.5,1.3,1.3,1.1,1.1,1.2,0.8,1,0.9,1)$ and $\rho_k=1$ for $k>10$, so that the two distributions differ in both their means and their variances.
The penalty $\kappa$ is selected by Equation~\eqref{eq:mse-selection} from the grid
\begin{equation*}
\mathcal{K}=\{0\} \cup 50\cdot\{1,2,5,10\} \cup 10^3\cdot\{1,2,3,5\} \cup 10^4\cdot\{1,2,5,7\} \cup 10^5\cdot\{1,2,5\},
\end{equation*}
whose lower boundary $\kappa=0$ corresponds to ordinary IPW.

\noindent\textbf{Comparison methods.} For the component-truncated IPW comparison method, we compute the density ratio retaining only the first five components. We also compare two ways of capping the ordinary IPW weights. For a threshold $c$, a weight-capping method replaces each unnormalized ordinary weight $w_t$ by $\min(w_t,c)$ and then normalizes these modified weights to sum to one. The first capping method uses a square-root weight cap with threshold $c=\bar w\sqrt{499}$, where $\bar w$ is the mean unnormalized weight and 499 is the number of retained transitions; this comparison is inspired by truncated importance sampling \citep{ionides2008truncated}.
The second capping method is an ESS-matched weight cap that chooses $c$ separately in each replicate to match the ESS of the $\kappa$-regularized estimator in that replicate; if that ESS does not exceed the ESS of ordinary IPW, we set $c=\infty$ so that the weights are unchanged.

\subsection{Confounding Adjustment Experiment}\label{app:sim-confounding}
For the confounding adjustment experiment of Section~\ref{sec:sim-confounding}, we take $\Kx=20$ in Equation~\eqref{eq:sim-scm} and take the spatial basis $\Vx$ to be the first 20 precipitation anomaly principal components described in Section~\ref{sec:sim-data}, with eigenvalues $\lambda^{\mathrm{pr}}_1,\ldots,\lambda^{\mathrm{pr}}_{20}$. We set $\boldsymbol\Phi_x=0.3\mathbf I_{20}$ and $\boldsymbol\Sigma_x=(1-0.3^2)\operatorname{diag}(\lambda^{\mathrm{pr}}_1,\ldots,\lambda^{\mathrm{pr}}_{20})$, so that the stationary variances of the confounder scores equal the GCM-derived eigenvalues. In all replicates, we set $\PhiTwo$ to be the least-squares coefficient matrix obtained by regressing annual piControl temperature scores on precipitation scores of the same year. The remaining parameters of $f_0$ and $f_1$ are as in Section~\ref{app:sim-design}.
The confounding strength $\gamma_{\mathrm{pr}}$ takes 21 equally spaced values,
\begin{equation*}
\gamma_{\mathrm{pr}}\in\{0,\gamma_{\max}/20,2\gamma_{\max}/20,\ldots,\gamma_{\max}\},
\qquad \gamma_{\max}=\simExpTwoGammaMax,
\end{equation*}
and $\gamma_{\mathrm{tas}}$ and the outcome noise are as in Equation~\eqref{eq:sim-outcome}.

\section{Application Details and Diagnostics}\label{app:diagnostics}

This section collects the supplementary material for the application of Section~\ref{sec:application}: the variables, observational record, and historical ensemble (Section~\ref{app:variables}); the extension of the intervention trajectory to 2025 (Section~\ref{app:ssp-pathway}); the application basis (Section~\ref{app:eof-stability}); the validation of the propensity and intervention models (Section~\ref{app:model-validation}); the penalty grid, together with the effective sample size and conditioning at the selected penalty $\kappa=\kappaPresent$ (Section~\ref{app:ess}); the land and ocean masks (Section~\ref{app:land-sea-mask}); the cropland mask and named agricultural regions (Section~\ref{app:cropland}); and representative weighted months (Section~\ref{app:weight-diagnostics}).

\subsection{Variables, Observational Record, and Historical Ensemble}\label{app:variables}

The observational record used in the application consists of the monthly-mean fields of the 851-year (\nMonthsPi-month) GISS-E2-1-G piControl run \texttt{r1i1p1f1} described in Section~\ref{sec:sim-data}. The treatment is near-surface air temperature (CMIP6 variable \texttt{tas}), the confounder is precipitation (\texttt{pr}), and the outcome, vapor-pressure deficit, is computed from \texttt{tas} and near-surface specific humidity (\texttt{huss}). 

Vapor-pressure deficit (hPa) at gridcell $s$ is
\begin{gather*}
\mathrm{VPD}(s) = \max\{e_s(T(s)) - e_a(s),\, 0\},\\
e_s(T) = 6.112\exp\!\left(\frac{17.67\,T}{T + 243.5}\right),
\qquad
e_a(s) = \frac{q(s)\,p}{0.622 + 0.378\,q(s)},
\end{gather*}
where $T(s)$ is \texttt{tas} in degrees Celsius, $e_s(T)$ is the saturation vapor pressure at temperature $T$, $e_a(s)$ is the actual vapor pressure derived from the specific humidity $q(s)$ given by \texttt{huss}, and the surface pressure is held constant at $p = 1013.25$~hPa. We use raw VPD as the outcome and monthly anomalies for visualization.

The \nHistMembers{} historical ensemble members used to estimate $f_1$ (Section~\ref{sec:app-data}) are pooled over five physics-and-forcing variant groups, denoted by variant IDs \texttt{p1f1}, \texttt{p1f2}, \texttt{p1f3}, \texttt{p3f1}, and \texttt{p5f1}, where the \texttt{p} index distinguishes treatments of atmospheric composition and aerosols, and the \texttt{f} index distinguishes prescribed forcing datasets.

\subsection{Extension to 2025 through SSP Projections}\label{app:ssp-pathway}

The Shared Socioeconomic Pathway (SSP) projections are simulations driven by prescribed future greenhouse gas concentrations \citep{meinshausen2020shared} and aerosol emissions, available for several pathways, including SSP1-2.6 (strong mitigation), SSP2-4.5 (intermediate), and SSP5-8.5 (high emissions).
To extend the record from 2015 to 2025 (Section~\ref{sec:app-data}), we use the intermediate SSP2-4.5 scenario; over this short horizon the three pathways give similar estimates (Table~\ref{tab:ghg-similarity} and Figure~\ref{fig:ssp-pathway}).
The \nSspMembers{} SSP2-4.5 scenario members are used to extend the intercept trajectory $\boldsymbol{\mu}_1(t^*)$ of Equation~\eqref{eq:intervention-dist} over intervention times $t^*$ from 2015 to 2025; we assume that the other parameters of $f_1$ are well represented by the historical simulations, and these are not re-estimated.

Because only the physics-and-forcing variants \texttt{p1f2}, \texttt{p3f1}, and \texttt{p5f1} of GISS-E2-1-G (Section~\ref{app:variables}) are present in both the historical and the scenario ensembles, we match each of the \nSspMembers{} SSP2-4.5 members to the corresponding historical variant, compute for each variant the difference between its mean scenario trajectory and its historical 2010--2014 mean, average these differences over the three variants with equal weight, and add the result to the \nHistMembers{}-member historical mean over 2010--2014. Measuring the scenario change against the matched historical mean, rather than against the full-ensemble mean, avoids a discontinuity in the trajectory at the transition from 2014 to 2015.

Over 2015--2025, the three pathways are similar in their greenhouse gas concentrations, and all remain close to the observed concentrations. Table~\ref{tab:ghg-similarity} reports, for SSP1-2.6, SSP2-4.5, and SSP5-8.5, the root-mean-square difference over 2015--2025 between the pathway's changes in CO$_2$, CH$_4$, and N$_2$O concentrations since 2014 and the observed changes in the NOAA global annual means \citep{noaa_gml_trends}. The pathway concentrations are the CMIP6 global means of \citet{meinshausen2020shared}, as distributed by \citet{nicholls2020reduced}; measuring both series as changes from their own 2014 values removes the constant offset between the two global-mean definitions. The table also combines the three gases into a change in radiative forcing with the simplified expressions of \citet{myhre1998new}. Relative to the observed forcing change of \ghgForcingObs~W\,m$^{-2}$ over 2014--2025, every pathway differs by at most \ghgRmseForcingMax~W\,m$^{-2}$ in root-mean-square terms, and SSP2-4.5 by \ghgRmseForcingMain~W\,m$^{-2}$.

Figure~\ref{fig:ssp-pathway} repeats the global estimate over 2014--2025 with the SSP1-2.6 and SSP5-8.5 extensions in place of SSP2-4.5, constructed in the same way from the same three variants and retaining the SSP2-4.5 penalty and the same bootstrap resamples. Over this period, the forced global-mean temperature of the three trajectories differs by at most \sspForcedTasMaxDiff~K. The three percentile intervals have a common intersection in every year, although the point estimates vary slightly across pathways.

\begin{table}[tbp]
\centering
\fontsize{10}{12}\selectfont
\setlength{\tabcolsep}{3pt}
\caption{Similarity of the SSP concentration pathways to observed greenhouse gas concentrations over 2015--2025. The four root-mean-square (RMS) columns give the RMS difference between each pathway's change since 2014 and the observed change, for each gas and for the combined forcing computed with the simplified expressions of \citet{myhre1998new}. The last column gives the forcing change from 2014 to 2025, with the observed value in the final row. Observations are NOAA Global Monitoring Laboratory global annual means, and pathway concentrations are the CMIP6 global means of \citet{meinshausen2020shared}.}
\label{tab:ghg-similarity}
\begin{tabular*}{\textwidth}{@{\extracolsep{\fill}}lrrrrr@{}}
\hline
 & \multicolumn{4}{c}{RMS difference from observed change, 2015--2025} & Forcing change \\
Pathway & CO$_2$ (ppm) & CH$_4$ (ppb) & N$_2$O (ppb) & Forcing (W\,m$^{-2}$) & 2014--2025 (W\,m$^{-2}$) \\
\hline
SSP1-2.6 & 1.7 & 37.4 & 2.0 & 0.013 & 0.426 \\
SSP2-4.5 & 2.2 & 16.0 & 1.5 & 0.028 & 0.480 \\
SSP5-8.5 & 3.4 & 12.6 & 1.3 & 0.041 & 0.514 \\
Observed & & & & & 0.443 \\
\hline
\end{tabular*}

\end{table}

\begin{figure}[tbp]
\centering
\includegraphics[width=0.7\textwidth]{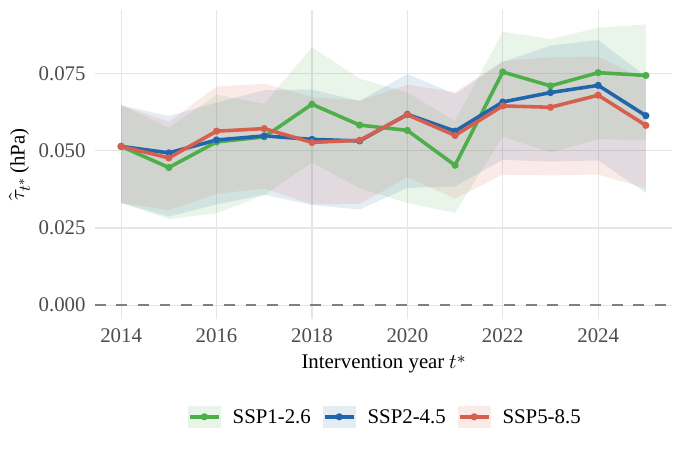}
\caption{Estimates of the global-mean effect on VPD (hPa) and 95\% block-bootstrap percentile intervals over intervention years 2014--2025 under the SSP1-2.6, SSP2-4.5, and SSP5-8.5 extensions of the historical trajectory, with $\kappa=\kappaPresent$.}
\label{fig:ssp-pathway}
\end{figure}

\subsection{Application Basis}\label{app:eof-stability}

The application basis, the hybrid basis of Section~\ref{sec:app-estimation}, combines the two forced-response patterns (the leading principal components of the annual ensemble-mean historical temperature anomalies) with orthogonalized piControl principal components, and is shared by $f_{\mathrm{obs}}$, $f_0$, and $f_1$. Figure~\ref{fig:application-basis} shows its first six vectors. The confounder basis $\Vx$ consists of the leading \KdimPr{} principal components of the deseasonalized piControl precipitation anomaly field, after standardization at each gridcell.

\begin{figure}[tbp]
\centering
\includegraphics[width=\textwidth]{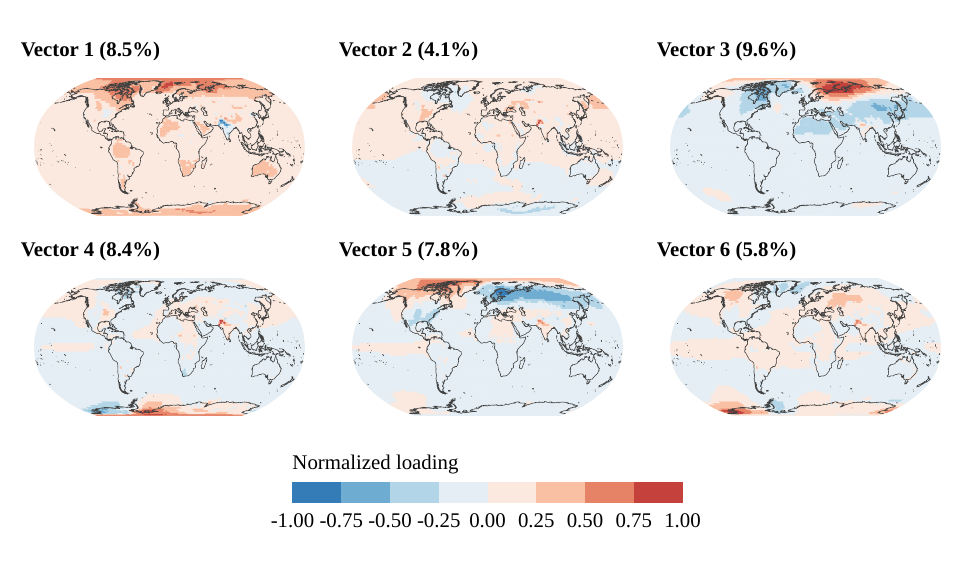}
\caption{The first six vectors of the application basis. Vectors 1 and 2 are the two forced-response patterns, and vectors 3--6 are the leading piControl principal components after orthogonalization. Each vector is divided by its largest absolute loading, and the percentage in each title is the share of piControl score variance along that vector.}
\label{fig:application-basis}
\end{figure}

We compare the application basis with two alternatives of the same dimension: the control basis, consisting of the leading \Kdim{} principal components of the piControl anomalies, and the historical basis, consisting of the leading \Kdim{} principal components of the historical anomalies. Using the five folds of historical members of Section~\ref{app:model-validation}, we construct the hybrid basis from the members outside each fold, with the two forced-response patterns computed from their annual ensemble mean, and the historical basis from the stacked monthly anomalies of the same members, relative to the piControl climatology. We then measure the share of the area-weighted squared norm of held-out anomaly fields that lies in the span of each basis. Table~\ref{tab:basis-validation} reports four comparisons: the annual ensemble-mean anomaly of the held-out members over 1950--2014, which estimates the forced response; its warming pattern, the 2005--2014 mean minus the 1950--1959 mean; the monthly anomalies of the individual held-out members; and, as an out-of-period target, the annual ensemble-mean SSP2-4.5 anomaly over 2015--2025, with all bases constructed from all historical members. The hybrid basis captures \basisForcedHybrid\% of the held-out forced response, compared with \basisForcedControl\% for the control principal components and \basisForcedHistorical\% for the historical principal components, and \basisSspHybrid\% of the scenario anomaly, compared with \basisSspControl\% and \basisSspHistorical\%. Replacing two control components with the two forced-response patterns therefore recovers nearly all of the advantage of a basis estimated from the historical ensemble, while the remaining directions describe the piControl variability used in the observational model.

\begin{table}[tbp]
\centering
\fontsize{10}{12}\selectfont
\setlength{\tabcolsep}{3pt}
\caption{Percentage of the area-weighted squared norm of held-out temperature anomaly fields captured by the span of each \Kdim-dimensional basis. The first three columns pool five folds of held-out historical members (forced-response patterns and historical principal components estimated without the held-out members); the last column uses the annual ensemble-mean SSP2-4.5 anomaly over 2015--2025. All anomalies are relative to the piControl monthly climatology. The basis used in the application is shown in bold.}
\label{tab:basis-validation}
\begin{tabular*}{\textwidth}{@{\extracolsep{\fill}}lrrrr@{}}
\hline
 & \multicolumn{3}{c}{Held-out historical members} & SSP2-4.5 \\
Basis & Forced response & Warming pattern & Monthly anomalies & 2015--2025 \\
\hline
\textbf{Hybrid (application)} & \textbf{96.7} & \textbf{99.0} & \textbf{78.2} & \textbf{99.2} \\
Control basis & 69.9 & 72.9 & 72.5 & 71.6 \\
Historical basis & 97.1 & 97.6 & 79.4 & 99.2 \\
\hline
\end{tabular*}

\end{table}

\subsection{Propensity and Intervention Model Validation}\label{app:model-validation}

We assess the specifications of the control intervention $f_0$ and the forced intervention $f_1$ by out-of-sample cross-validation, holding the spatial bases and preprocessing fixed. Both models are compared across the same five specifications, which differ in whether the transition matrix, the intercept, and the residual covariance are shared across calendar months or month-specific. For $f_1$, we add a sixth specification, in which the transition matrix is fixed at the control estimate $\widehat{\boldsymbol\Phi}_0$, so that the intervention changes only the conditional mean level and the residual covariance. Table~\ref{tab:model-validation} reports the mean held-out Gaussian log density per monthly transition of the \Kdim-dimensional score vector for each specification. Because the observational record is the piControl run, $f_0=f_{\mathrm{obs}}$, and the model for $f_0$ is also the propensity model.

\noindent\textbf{Control intervention.} The model for $f_0$ is assessed by cross-validation over the 851 piControl years with ten interleaved folds: year $y$ belongs to fold ${((y-1)\bmod 10)+1}$, so that the folds do not overlap and each contains every tenth year, and each transition belongs to the fold of its target year. Each fold is predicted from a model fitted to the other nine, with the lagged treatment scores and the precipitation scores as predictors.

\noindent\textbf{Forced intervention.} The intercept trajectory $\boldsymbol\mu_1(t^{*})$ of Equation~\eqref{eq:intervention-dist} is estimated as the annual mean of the ensemble-mean treatment scores in year $t^{*}$, centered over 1950--2014. The transition matrix $\boldsymbol\Phi_1$ and the month-specific intercepts $\boldsymbol\mu_{1,m}$ are then estimated by least-squares regression of the pooled member scores, with this trajectory subtracted, on the lagged treatment scores, so that $\boldsymbol\mu_{1,m}(t^{*}) = \boldsymbol\mu_{1,m} + \boldsymbol\mu_1(t^{*})$, and each residual covariance $\boldsymbol\Sigma_{\varepsilon,1,m}$ is the sample covariance of the regression residuals in calendar month $m$. The model for $f_1$ is assessed by leave-members-out cross-validation over the \nHistMembers{} historical members, which are divided into five folds. For each fold, the forced trajectory and the regression are fitted to the members outside the fold, and the monthly transitions of the held-out members are scored. 

Among the specifications with a shared transition matrix, month-specific residual covariances account for nearly all of the improvement over the specification with no seasonal terms, while month-specific intercepts alone give little or no improvement. The specification with month-specific transition matrices (a periodic VAR) attains the highest held-out log density for both models, but at the cost of multiplying the number of persistence parameters twelvefold, from $\Kdim^2$ to $12\cdot\Kdim^2$. The application therefore uses a shared transition matrix with month-specific intercepts and residual covariances for both $f_0$ and $f_1$. For $f_1$, fixing the transition matrix at $\widehat{\boldsymbol\Phi}_0$ lowers the held-out log density relative to estimating a shared $\boldsymbol\Phi_1$ from the historical ensemble.

\begin{table}[tbp]
\centering
\fontsize{10}{12}\selectfont
\setlength{\tabcolsep}{3pt}
\caption{Mean held-out Gaussian log density per monthly transition for the control intervention $f_0$ and the forced intervention $f_1$. The rows differ in whether the transition matrix, intercept, and residual covariance are shared across calendar months (shared) or month-specific (monthly); higher values indicate better out-of-sample fit. The specification used in the application is shown in bold. Fixing the transition matrix at the control estimate $\widehat{\boldsymbol\Phi}_0$ applies only to $f_1$.}
\label{tab:model-validation}
\begin{tabular*}{\textwidth}{@{\extracolsep{\fill}}lllrr@{}}
\hline
\multicolumn{3}{@{}l}{Specification} & \multicolumn{2}{r@{}}{Held-out log density} \\
Transition matrix & Intercept & Residual covariance & $f_0$ & $f_1$ \\
\hline
shared & shared & shared & $-200.76$ & $-209.26$ \\
shared & monthly & shared & $-200.82$ & $-209.03$ \\
shared & shared & monthly & $-197.05$ & $-202.86$ \\
\textbf{shared} & \textbf{monthly} & \textbf{monthly} & {\boldmath$-197.11$} & {\boldmath$-202.60$} \\
monthly & monthly & monthly & $-193.62$ & $-201.23$ \\
$\widehat{\boldsymbol\Phi}_0$ & monthly & monthly & --- & $-203.32$ \\
\hline
\end{tabular*}

\end{table}

\subsection{Penalty Grid, Effective Sample Size, and Conditioning}\label{app:ess}

In the application, the penalty $\kappa$ is selected by Equation~\eqref{eq:mse-selection} from the grid
\begin{equation*}
\mathcal{K}=\{0\} \cup 10^4\cdot\{1,2,3,5,7\} \cup 10^5\cdot\{1,2,5\},
\end{equation*}
whose lower boundary $\kappa=0$ corresponds to ordinary IPW. The selected value $\widehat\kappa=\kappaPresent$ is used for the diagnostics in the remainder of this section and in Sections~\ref{app:land-sea-mask}--\ref{app:weight-diagnostics}.

Figure~\ref{fig:ess-year} shows the effective sample size of the adjusted weights at each intervention year. The ESS declines in later intervention years, as the forced distribution moves farther from the control distribution, and is \essPresentJoint{} at $t^{*}=2025$; without precipitation adjustment it is \essPresentMarginal{} in that year.

\begin{figure}[tbp]
\centering
\includegraphics[width=\textwidth]{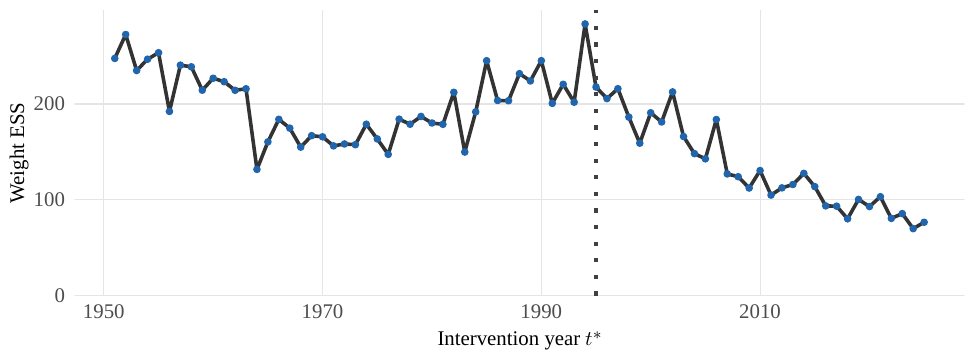}
\caption{Effective sample size of the adjusted weights, $1/\sum_t\alpha_t^2$, by intervention year $t^{*}$ at $\kappa=\kappaPresent$, where $\alpha_t=w_t^{(1;t^*)}(\kappa)\big/\{12\sum_{t'\in\mathcal{J}_m}w_{t'}^{(1;t^*)}(\kappa)\}$, for $t$ in the set $\mathcal{J}_m$ of transitions from calendar month $m$, are the normalized weights of the monthly averaged estimator of Section~\ref{sec:app-estimation}. The vertical dotted line marks the emergence year \emergenceYear{} of Figure~\ref{fig:intervention-year}.}
\label{fig:ess-year}
\end{figure}

Each historical monthly residual covariance $\boldsymbol\Sigma_{\varepsilon,1,m}$ is estimated from at least \sigmaHistNresid{} pooled residuals. Across calendar months, its condition number ranges from \sigmaHistCondMin{} to \sigmaHistCondMax{} before penalization and from \sigmaHistCondPenMin{} to \sigmaHistCondPenMax{} at the selected penalty, so the penalty of Section~\ref{sec:propensity} substantially reduces the condition number, consistent with its purpose of reducing the contribution of low-variance directions to the weights.

\subsection{Land and Ocean Masks}\label{app:land-sea-mask}

Gridcells are classified as land or ocean by the $1^\circ$ land--sea mask distributed with the NCAR Command Language \citep{ncl2019}, which is derived from RAND's global elevation and depth data \citep{rand_elevation}. The mask is assigned to the $2.5^\circ$ analysis grid by nearest neighbor, with lakes, small islands, and ice shelves counted as land. Land and ocean estimates are cosine-latitude-weighted means over the gridcells of each region, computed on the original record and on every bootstrap replicate with the same weights, so that the global estimate is the area-weighted combination of the land and ocean estimates. Land gridcells make up \landAreaPct\% of the weighted area and account for \landSharePresentPct\% of the global estimate at $t^{*}=2025$.

\subsection{Cropland Mask and Named Agricultural Regions}\label{app:cropland}

Cropland fractions are taken from the LUH2 land-use states for \cropYear{} \citep{hurtt2020harmonization}, the land-use forcing of the CMIP6 historical simulations, including those of GISS-E2-1-G. Cropland in each $0.25^\circ$ LUH2 gridcell is the sum of the fractions of its five crop states (C3 and C4 annual crops, C3 and C4 perennial crops, and C3 nitrogen-fixing crops), multiplied by the gridcell area. Each $2.5^\circ$ analysis gridcell contains exactly $10\times10$ LUH2 gridcells, and its cropland fraction is their total cropland area divided by their total area (Figure~\ref{fig:cropland-regions}, Panel~A). The cropland shares of Section~\ref{sec:app-adaptation} weight each gridcell by its area times its cropland fraction. Panel~B of Figure~\ref{fig:adjustment-comparison} displays sign-reversed gridcells with at least 2\% cropland, with opacity increasing with the cropland fraction and saturating at 60\%. Figure~\ref{fig:cropland-zooms} enlarges the four highlighted regions; in that figure, darker shading indicates a larger absolute difference between the adjusted and unadjusted estimates (hPa). The 2\% cropland threshold is used only for visualization and does not enter the regional estimates.

Table~\ref{tab:crop-classes} gives the share and area of global cropland in each of seven adjustment classes: for a positive (drying) unadjusted estimate, adjustment increases its magnitude (DI), decreases its magnitude (DD), or reverses its sign (DR); for a negative (moistening) unadjusted estimate, the corresponding classes are MI, MD, and MR, respectively. The seventh class comprises cropland on which both estimates are negligible (NG). In total, adjustment reverses the sign of the estimate over \adjClassCropRevPct\% of cropland, split roughly evenly between drying-to-moistening reversals (\adjClassCropPosRevMha{} million hectares, Mha) and moistening-to-drying reversals (\adjClassCropNegRevMha{}~Mha).

To summarize the regional effects of adjustment on cropland, Table~\ref{tab:crop-regions} reports the cropland of fifteen named agricultural regions by adjustment class, with gridcells assigned to a region by their centers; Figure~\ref{fig:cropland-regions} maps the regions and displays the same breakdown. The first four regions contain \cropWindowsRevPct\% of the sign-reversed cropland and are highlighted in Figure~\ref{fig:adjustment-comparison}. The remaining \cropOtherN{} are further major cropland regions, which together contain \cropOtherMha{}~Mha of cropland and \cropOtherRevMha{}~Mha of sign-reversed cropland. Among them, adjustment intensifies the drying estimate over most of the cropland of the Indo-Gangetic Plain and the Deccan Plateau, and decreases it over most of the cropland of the Sahel and the Brazilian Cerrado. In the other seven regions, both estimates lie within two bootstrap standard errors of zero over most of the cropland.

\begin{table}[tbp]
\centering
\caption{Global cropland by adjustment class at $t^{*}=2025$, as a percentage of global cropland and in Mha. Non-negligible gridcells are grouped by the sign of the unadjusted estimate; a gridcell is considered negligible when both estimates lie within two bootstrap standard errors of zero. Totals may differ from the sums of displayed components because of rounding.}
\label{tab:crop-classes}
\begin{tabular}{@{}llrr@{}}
\hline
Class & Code & Cropland (\%) & Cropland (Mha) \\
\hline
Unadjusted estimate positive (drying) &  & 30.9 & 488 \\
\quad Same sign, magnitude increased & DI & 11.3 & 179 \\
\quad Same sign, magnitude decreased & DD & 15.7 & 249 \\
\textbf{\quad Sign reversed} & \textbf{DR} & \textbf{3.8} & \textbf{60} \\
Unadjusted estimate negative (moistening) &  & 9.8 & 155 \\
\quad Same sign, magnitude increased & MI & 2.2 & 35 \\
\quad Same sign, magnitude decreased & MD & 3.5 & 55 \\
\textbf{\quad Sign reversed} & \textbf{MR} & \textbf{4.1} & \textbf{65} \\
Negligible & NG & 59.3 & 939 \\
\hline
All cropland &  & 100.0 & 1582 \\
\hline
\end{tabular}

\end{table}

\begin{table}[tbp]
\centering
\caption{Named cropland regions: bounds and cropland in \cropYear{} (Mha) by adjustment class at $t^{*}=2025$. The first four regions contain most of the sign-reversed cropland and are highlighted in Figure~\ref{fig:adjustment-comparison}; the regions below are additional major cropland regions, in decreasing order of cropland. Column codes are those of Table~\ref{tab:crop-classes}; the sign-reversed classes DR and MR are shown in bold. Totals may differ from the sums of displayed components because of rounding.}
\label{tab:crop-regions}
\small
\setlength{\tabcolsep}{3pt}
\begin{tabular}{@{}lllrrrrrrrr@{}}
\hline
 & & & \multicolumn{8}{c@{}}{Cropland (Mha)} \\
Region & Latitude & Longitude & Total & DI & DD & \textbf{DR} & MI & MD & \textbf{MR} & NG \\
\hline
Northern Great Plains & 44\textdegree{}N--56\textdegree{}N & 115\textdegree{}W--97\textdegree{}W & 50 & 0 & 0 & \textbf{0} & 0 & 23 & \textbf{22} & 4 \\
East African highlands & 10\textdegree{}S--10\textdegree{}N & 29\textdegree{}E--42\textdegree{}E & 46 & 8 & 1 & \textbf{0} & 0 & 0 & \textbf{30} & 7 \\
Argentine Pampas & 40\textdegree{}S--28\textdegree{}S & 65\textdegree{}W--55\textdegree{}W & 43 & 0 & 0 & \textbf{11} & 12 & 2 & \textbf{0} & 18 \\
Southeastern Australia & 38\textdegree{}S--26\textdegree{}S & 137\textdegree{}E--153\textdegree{}E & 28 & 0 & 8 & \textbf{14} & 0 & 0 & \textbf{0} & 7 \\
\hline
Indo-Gangetic Plain & 21\textdegree{}N--32\textdegree{}N & 70\textdegree{}E--90\textdegree{}E & 127 & 74 & 36 & \textbf{0} & 0 & 0 & \textbf{0} & 18 \\
Ukraine and southern Russia & 44\textdegree{}N--55\textdegree{}N & 25\textdegree{}E--50\textdegree{}E & 86 & 0 & 0 & \textbf{0} & 0 & 0 & \textbf{0} & 86 \\
US Corn Belt & 36\textdegree{}N--48\textdegree{}N & 97\textdegree{}W--80\textdegree{}W & 69 & 0 & 0 & \textbf{0} & 6 & 7 & \textbf{0} & 55 \\
Sahel & 10\textdegree{}N--16\textdegree{}N & 15\textdegree{}W--20\textdegree{}E & 68 & 2 & 53 & \textbf{0} & 0 & 0 & \textbf{0} & 13 \\
Western and central Europe & 43\textdegree{}N--55\textdegree{}N & 5\textdegree{}W--20\textdegree{}E & 62 & 0 & 0 & \textbf{3} & 0 & 1 & \textbf{0} & 58 \\
North China Plain & 32\textdegree{}N--42\textdegree{}N & 110\textdegree{}E--123\textdegree{}E & 42 & 0 & 0 & \textbf{0} & 1 & 0 & \textbf{0} & 40 \\
Deccan Plateau & 15\textdegree{}N--21\textdegree{}N & 73\textdegree{}E--82\textdegree{}E & 31 & 29 & 0 & \textbf{0} & 0 & 0 & \textbf{0} & 2 \\
Mainland Southeast Asia & 10\textdegree{}N--20\textdegree{}N & 98\textdegree{}E--108\textdegree{}E & 29 & 0 & 0 & \textbf{0} & 0 & 0 & \textbf{0} & 29 \\
Kazakh steppe & 48\textdegree{}N--56\textdegree{}N & 60\textdegree{}E--80\textdegree{}E & 24 & 0 & 0 & \textbf{0} & 0 & 2 & \textbf{2} & 20 \\
Brazilian Cerrado & 20\textdegree{}S--10\textdegree{}S & 55\textdegree{}W--43\textdegree{}W & 24 & 1 & 23 & \textbf{0} & 0 & 0 & \textbf{0} & 0 \\
Southern Great Plains & 30\textdegree{}N--40\textdegree{}N & 105\textdegree{}W--98\textdegree{}W & 18 & 0 & 0 & \textbf{0} & 0 & 0 & \textbf{0} & 18 \\
\hline
\end{tabular}

\end{table}

\begin{figure}[tbp]
\centering
\includegraphics[width=\textwidth]{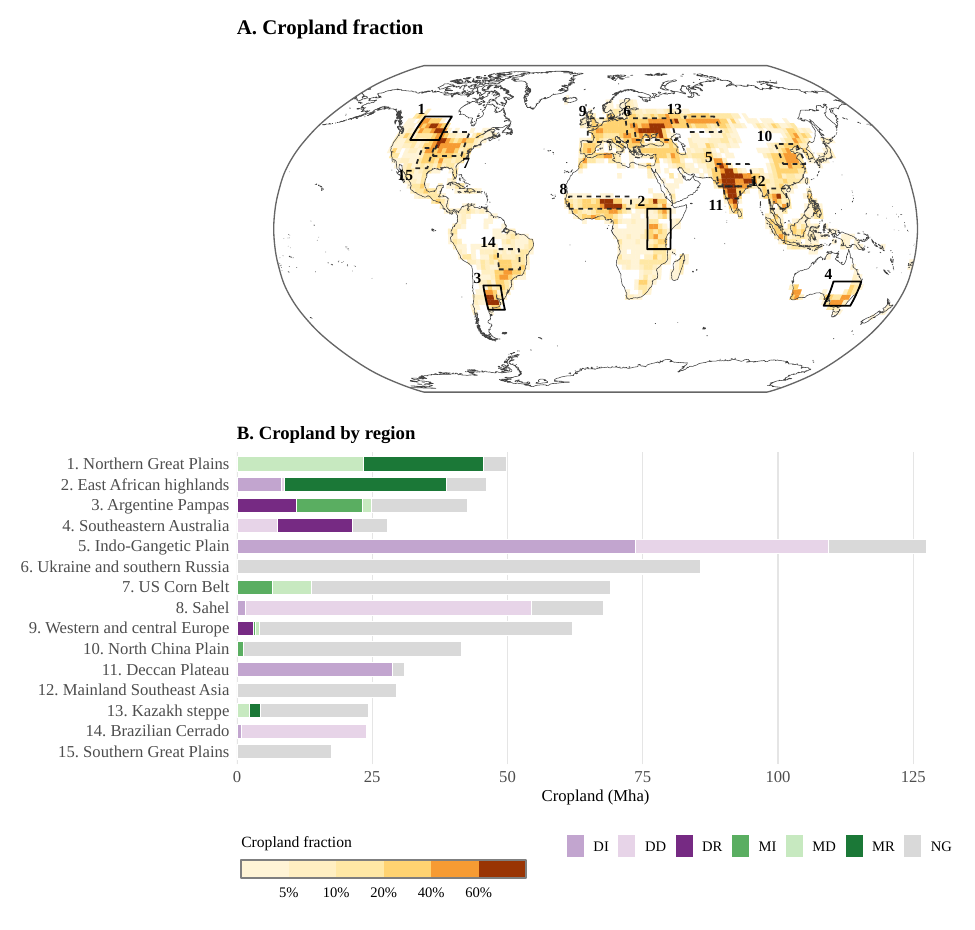}
\caption{Cropland mask and named regions. Panel~A shows the LUH2 cropland fraction of each analysis gridcell in \cropYear{}, with the fifteen named regions of Table~\ref{tab:crop-regions} numbered in table order; solid boxes outline the four regions containing most of the sign-reversed cropland and dashed boxes the further regions. Panel~B shows the cropland of each region by adjustment class at $t^{*}=2025$, using the codes of Table~\ref{tab:crop-classes}: purple shades represent drying and green shades represent moistening for the unadjusted estimate, with the sign-reversed classes darkest and the negligible class (NG) in gray.}
\label{fig:cropland-regions}
\end{figure}

\begin{figure}[tbp]
\centering
\includegraphics[width=\textwidth]{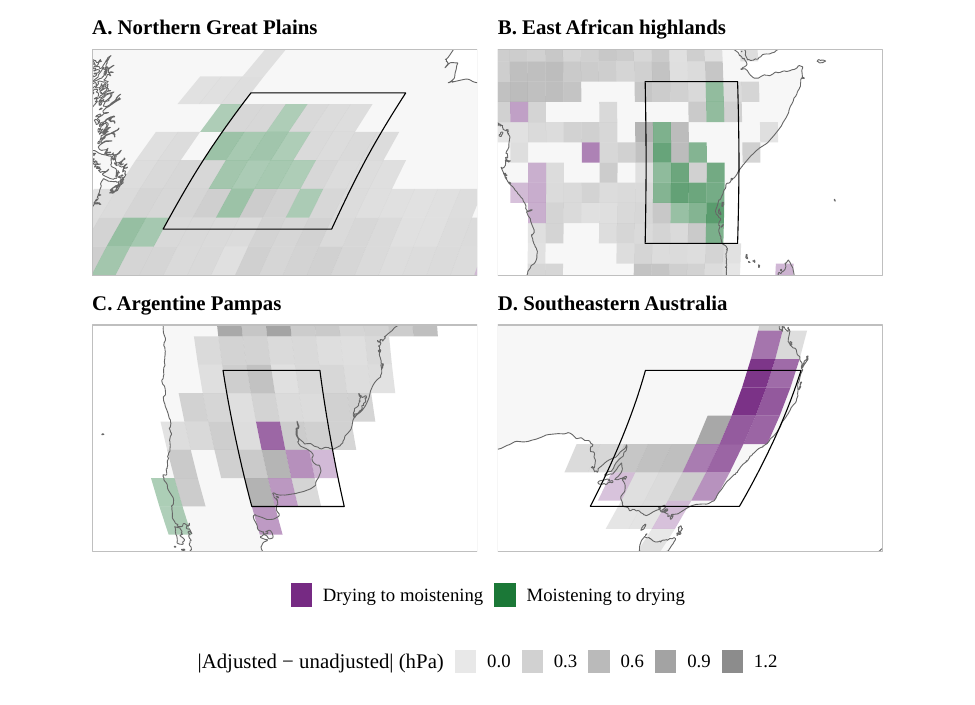}
\caption{Regional enlargements of the cropland sign reversals in Figure~\ref{fig:adjustment-comparison} at $t^{*}=2025$. Purple and green show reversals from drying to moistening and from moistening to drying, respectively; gray shows other gridcells with at least 2\% cropland. Darker shading indicates a larger absolute difference between the adjusted and unadjusted estimates (hPa), on a common scale across all four panels. Boxes denote the named regions of Table~\ref{tab:crop-regions}.}
\label{fig:cropland-zooms}
\end{figure}

\subsection{Representative Weighted Months}\label{app:weight-diagnostics}

Figure~\ref{fig:weight-draws} shows three piControl months evaluated under the regularized conditional densities at $\kappa=\kappaPresent$ and $t^{*}=2025$: a month with high density under the control distribution $f_0$, a month with high density under the forced distribution $f_1$, and a month with low density under both. The maps illustrate piControl fields at contrasting values of these conditional densities: the month with high density under the forced distribution has warm anomalies over most of the globe, strongest over Antarctica, whereas the month with low density under both distributions combines strong cold anomalies over the Arctic with warm anomalies over Antarctica.

\begin{figure}[tbp]
\centering
\includegraphics[width=\textwidth]{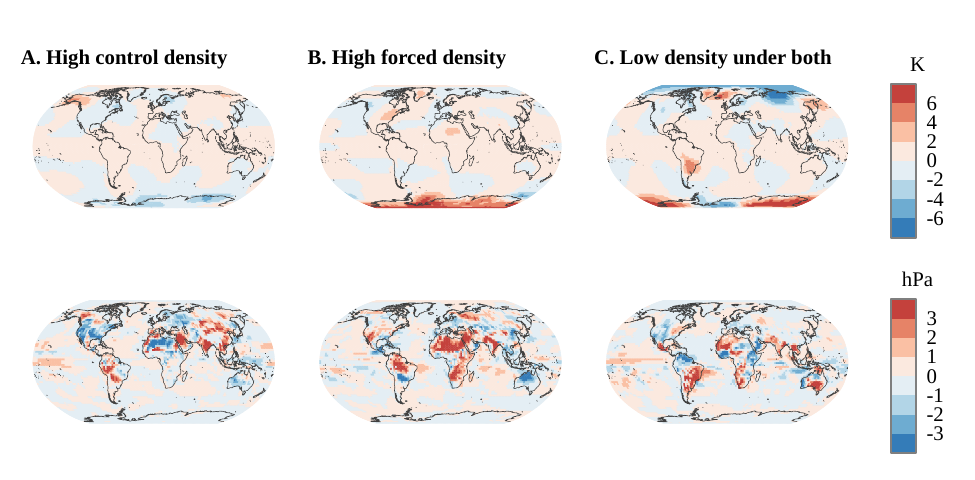}
\caption{Representative piControl months by density under the control and forced distributions. Columns: a month with high density under the control distribution, a month with high density under the forced distribution, and a month with low density under both (months are chosen as the fifth-most-extreme within each category so that a single outlier is not shown as representative). Top row: the near-surface air temperature anomaly reconstruction (K). Bottom row: the corresponding VPD anomaly (hPa).}
\label{fig:weight-draws}
\end{figure}

\clearpage

\end{document}